\documentclass[journal]{IEEEtran}

\usepackage{algorithm}
\usepackage{algorithmic}
\usepackage{amsmath,amssymb}
\usepackage{graphicx}
\usepackage{cite}
\newcommand{\bm}[1]{\mbox{\boldmath{$#1$}}}
\usepackage{stfloats}
\usepackage{hyperref}
\usepackage{amsthm}
\usepackage{subfigure}

\newcommand{\headcell}[1]{%
	\begin{tabular}[c]{@{}c@{}}#1\end{tabular}}

\newcommand{\headcellleft}[1]{%
	\begin{tabular}[c]{@{}l@{}}#1\end{tabular}}

\allowdisplaybreaks[4]
\newcommand{\C}{\mathbb{C}}
\newcommand{\R}{\mathbb{R}}
\newcommand{\tr}{\operatorname{Tr}}

\theoremstyle{plain}
\newtheorem{proposition}{Proposition}
\newtheorem{remark}{Remark}

\begin{document}

\title{Rotatable Antenna Relaying: Joint Precoding and Antenna Pointing Design}
\author{Ying~Gao, Qingqing~Wu, and Wen~Chen \vspace{-2mm} 
\thanks{The authors are with the School of Integrated Circuit, Shanghai Jiao Tong University, Shanghai 201210, China (e-mail: yinggao@sjtu.edu.cn; qingqingwu@sjtu.edu.cn; wenchen@sjtu.edu.cn).}}

\maketitle

\begin{abstract}
	This paper investigates a rotatable antenna (RA)-enhanced half-duplex amplify-and-forward relaying system, where a multi-antenna base station (BS) serves multiple single-antenna users via a multi-antenna relay. The BS and users employ isotropic antennas, whereas the relay employs directional RAs whose pointing matrix is shared by both hops to avoid inter-hop reorientation delay and control overhead. We aim to maximize the minimum signal-to-interference-plus-noise ratio (SINR) among all users by jointly optimizing the BS precoding, relay precoding, and RA pointing matrices. To tackle this non-convex problem, we first investigate the single-user scenario and reduce the joint design to RA pointing optimization using the optimal relay precoding matrix available in closed form. A manifold-aware Frank--Wolfe (MFW) method is then employed to obtain a suboptimal pointing solution. Under a symmetric far-field line-of-sight geometry, we further characterize the globally optimal pointing structure and derive a directivity threshold separating common pointing from antenna splitting. Building on this MFW procedure, we next address the general multiuser scenario. Specifically, the quadratic transform is first applied to obtain an equivalent auxiliary-variable formulation, which is subsequently solved suboptimally via alternating optimization (AO). In particular, we employ a safeguarded extension of the MFW method based on log-sum-exp smoothing to update the RA pointing matrix in each AO iteration. Simulation results demonstrate that the proposed algorithms consistently achieve the best signal-to-noise ratio (SNR) and SINR performance among all considered schemes. It is further shown that the optimized RA pointing adapts to the two-hop geometry, the preferred directivity factor depends on user load, and relay placement with relatively balanced two-hop propagation conditions is generally preferable. 
\end{abstract}


\vspace{-2mm}
\section{Introduction}
Wireless relaying is an effective approach to extending coverage and improving transmission reliability in infrastructure-based and ad hoc networks \cite{2004_Pabst_relay}. When the direct link is severely degraded by path loss, shadowing, or multipath fading, an intermediate relay can assist communication between the source and destination through two-hop transmission. Relay protocols are broadly classified as regenerative or non-regenerative \cite{2004_Laneman_relay}. Decode-and-forward (DF) is a representative regenerative protocol, in which the relay decodes and re-encodes the received signal before forwarding it. By contrast, amplify-and-forward (AF) is a common non-regenerative protocol that directly amplifies and forwards the received signal, thereby reducing processing delay and implementation complexity. Relays may also operate in full-duplex or half-duplex mode. Full-duplex operation enables simultaneous transmission and reception and can improve spectral efficiency, but imposes stringent self-interference cancellation and hardware requirements \cite{2014_Sabharwal_FD}. Half-duplex operation instead separates reception and transmission over orthogonal time or frequency resources. Consequently, half-duplex AF relaying combines low processing complexity with ease of implementation. 

The performance of relay networks can be further enhanced by multiple-input multiple-output (MIMO) technology, which exploits spatial degrees of freedom (DoFs) to increase spectral efficiency and mitigate channel fading \cite{2004_PAULRAJ_MIMO}. In multiuser relay systems, equipping the base station (BS) and relay with multiple antennas enables spatial multiplexing of multiple user streams and facilitates interference suppression through coordinated processing across the two hops. Accordingly, existing studies have analyzed the achievable performance of MIMO relay systems and developed various transceiver designs based on linear precoding and power allocation \cite{2008_Chae_relay,2009_Rui_relay,2009_Yue_relay}.  However, conventional MIMO implementations generally require large fixed-position antenna (FPA) arrays, resulting in increased hardware cost, power consumption, and signal processing complexity \cite{2020_Chowdhury_challenges}. Moreover, the fixed positions and orientations of FPAs restrict spatial channel sampling to discrete locations, preventing antenna arrays from adapting to continuous channel variations and thereby limiting the attainable diversity and multiplexing gains \cite{2023_Lipeng_overview}.

To overcome the limitations of FPAs, movable antennas (MAs) \cite{2023_Lipeng_Modeling}, also known as fluid antennas \cite{2021_Wong_fluid}, have emerged as a promising technology that enables flexible antenna positioning within continuous spatial regions. This positional flexibility provides additional DoFs, allowing MAs to exploit local channel variations to improve spectral efficiency, mitigate interference with fewer radio-frequency chains, and achieve full array gain with null steering \cite{2023_Wenyan_MIMO,2023_Lipeng_uplink,2023_Lipeng_null}. These advantages have motivated extensive research across various wireless systems \cite{2024_Zhenyu_uplink,2024_Ying_multicast,2025_Ziyuan_twotime,2024_Honghao_IFC,2025_Yichi_hybrid,2026_Ying_WPCN,2026_Ying_cost}. In particular, several studies have investigated MA-enhanced relay communications \cite{2025_Nianzu_relay,2025_Shihao_realy,2024_Ruopeng_relay}. For example, the authors of \cite{2025_Nianzu_relay} studied half-duplex DF and AF relaying between two single-antenna terminals, jointly optimizing the stage-dependent relay-side MA positions and, in the AF case, the relay beamforming matrix to maximize the achievable rate. Furthermore, a half-duplex AF two-way relay system was investigated in \cite{2025_Shihao_realy}, where the relay-side MA positions and beamforming matrix were jointly optimized to maximize the minimum signal-to-noise ratio (SNR) over the two transmission directions. 

Although MAs provide positional flexibility, their fixed orientations limit antenna reconfiguration to positional adjustments. To exploit both position and orientation DoFs, six-dimensional movable antennas (6DMAs) have been proposed to jointly adjust antenna positions and three-dimensional orientations \cite{2025_Xiaodan_6DMA,2026_Xiaodan_6DMA_survey}. As a lightweight and implementation-friendly variant of 6DMAs, rotatable antennas (RAs) maintain fixed positions while allowing the three-dimensional boresight of each directional antenna to be independently adjusted \cite{2026_Beixiong_RA,2026_Beixiong_RA_mag}. Such orientation control changes the antenna gains along different propagation paths, thereby enhancing desired signals and mitigating interference without translating the antenna elements. Compared with position-adjustable architectures, RAs avoid the additional movement regions and sliding mechanisms required for antenna translation, reducing mechanical complexity and facilitating integration with existing planar arrays. These advantages have motivated extensive studies on RA-enabled systems for multiuser communications \cite{2026_Guoying_RA,2026_Ailing_RA,2026_Xingxiang_RA_beyong}, spectrum sharing \cite{2026_Xingxiang_RA_spectrum,2026_Yanhua_RA}, cell-free communications \cite{2026_Kecheng_RA,2026_Xingxiang_RA_cell}, physical-layer security \cite{2025_Liang_RA_security}, and integrated sensing and communications \cite{2025_Chao_RA_ISAC,2025_Yunan_RA_ISAC,2026_Shiying_RA_ISAC}. In these studies, RA orientations are commonly optimized jointly with transmit or receive beamforming and other system variables. 

Despite these advances, research on RA-enhanced relay communications remains in its infancy. In half-duplex relaying, a key design choice is whether the relay-side RAs should be reconfigured between the two hops or retain the same orientations throughout a transmission block. Although per-hop reconfiguration offers greater flexibility, it may incur additional adjustment delay, energy consumption, and control overhead. Maintaining a shared RA pointing matrix across both hops therefore provides a practical low-overhead alternative. This architecture raises a fundamental question: \emph{How should the relay-side RAs be pointed to balance first-hop reception from the BS and second-hop transmission to spatially distributed users?} This question stems from the fact that strengthening the RA gain toward the BS may weaken the gains toward the users, and vice versa, while unequal user path losses further complicate this balance. The trade-off becomes more pronounced as antenna directivity increases. Answering this question is particularly challenging in AF relay systems, which are attractive for their low processing latency and implementation complexity. Unlike DF relaying, whose decoding operation separates the two-hop signal models, AF relaying forwards a linearly processed mixture of signals, interference, and relay noise. Consequently, the shared RA pointing matrix simultaneously affects both hop channels and the relay transmit power, while coupling with the BS precoding and relay precoding matrices in the fractional end-to-end signal-to-interference-plus-noise ratios (SINRs). These intertwined effects lead to a tightly coupled non-convex joint design problem and call for dedicated analytical and algorithmic treatment. 

\begin{figure}[!t]
	\centering
	\includegraphics[scale=0.48]{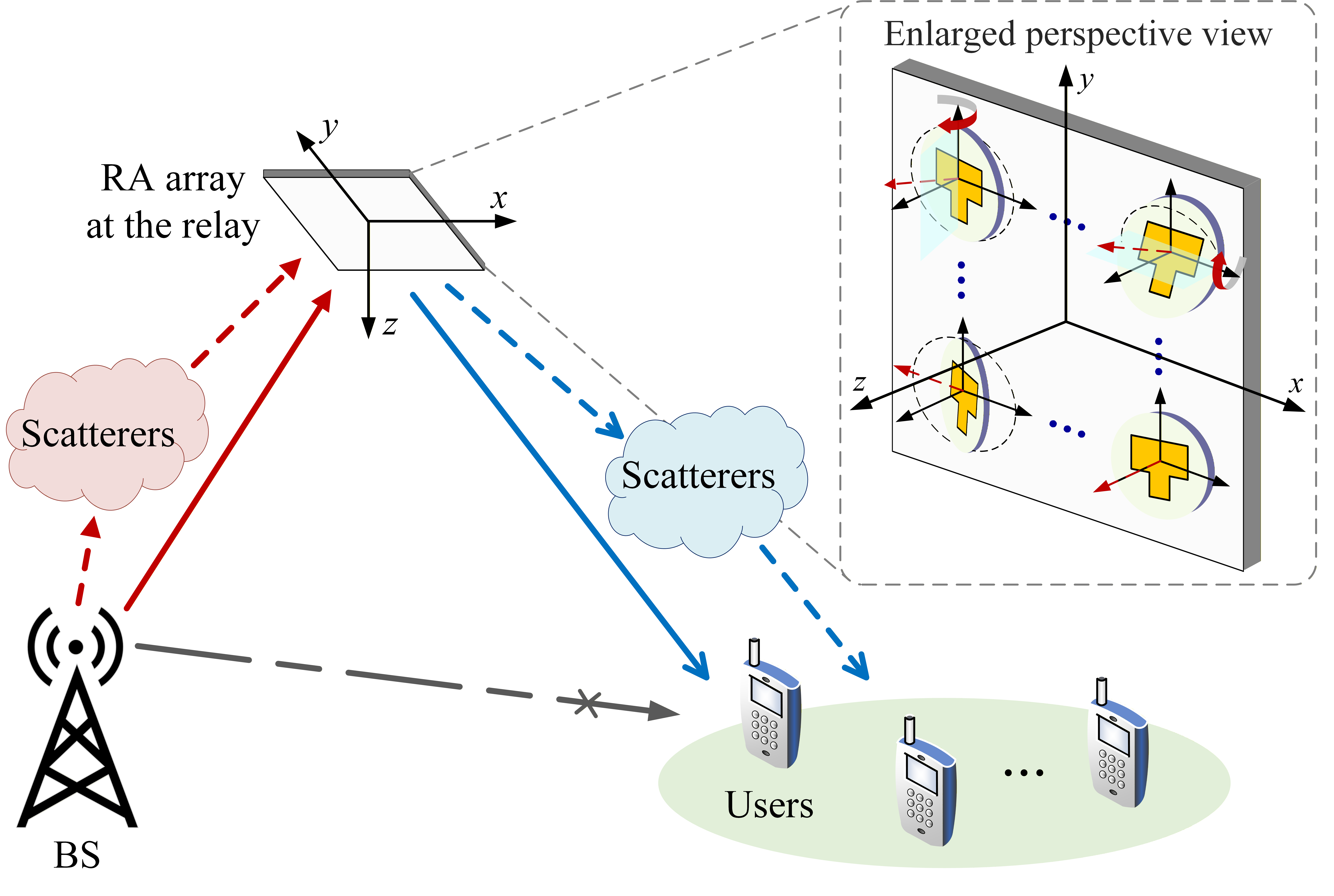}
	\caption{Illustration of an RA-enhanced multiuser relay system.}
	\label{fig:system_model}
	\vspace{-1mm}
\end{figure}

Motivated by the above discussions, we investigate an RA-enhanced half-duplex AF relaying system. As shown in Fig.~\ref{fig:system_model}, a multi-antenna relay assists a multi-antenna BS in serving multiple single-antenna users. The BS and users employ isotropic antennas, whereas the relay employs directional RAs. Considering the reorientation delay and control overhead, the relay-side RA pointing matrix is optimized once per block and shared by both hops. For user fairness, we jointly optimize the BS precoding, relay precoding, and RA pointing matrices to maximize the minimum SINR among all users, resulting in a non-convex problem. The proposed solutions and the associated findings constitute the main contributions of this paper, summarized as follows:
\begin{itemize}
	\item To lay the algorithmic foundation for the general multiuser scenario, we first study the simplified single-user case with a single-antenna BS. Leveraging the closed-form characterization of the optimal relay precoding matrix, we reduce the joint design to an RA pointing problem and obtain a suboptimal solution using a manifold-aware Frank--Wolfe (MFW) method. For the special case of symmetric far-field line-of-sight (LoS) geometry, we further characterize the globally optimal pointing structure and derive a directivity threshold separating common pointing from antenna splitting. 
	
	\item Building on the single-user MFW procedure, we next address the general multiuser scenario, where inter-user interference couples all design variables. Although alternating optimization (AO) is a natural choice for handling this coupling, applying it directly would retain fractional SINR constraints that complicate blockwise optimization. We therefore apply the quadratic transform to recast the problem into an equivalent auxiliary-variable formulation amenable to AO. Within each AO iteration, the RA pointing matrix is updated using a safeguarded MFW method based on log-sum-exp smoothing, while the BS precoding and relay precoding matrices are updated by solving convex subproblems.  
	
	\item Simulation results show that: 1) The proposed algorithms consistently achieve the best SNR/SINR performance among all considered schemes across the evaluated settings. 2) The transition from common pointing to antenna splitting characterized under symmetric far-field LoS geometry is also observed in asymmetric multipath channels, where the preferred split is unbalanced. In multiuser scenarios, the optimized RAs increasingly form distinct BS- and user-oriented groups as the directivity factor increases. 3) A higher directivity factor is favored under light user loads, whereas a moderate value better accommodates dispersed users under heavy loads. 4) A relay location that balances the propagation conditions of the two hops is generally preferable. 
\end{itemize} 

The remainder of this paper is organized as follows. Section~\ref{sec:system_model} presents the system model and formulates the max-min SINR problem. Section~\ref{sec:single_user} investigates the single-user scenario, developing an MFW algorithm and characterizing the optimal RA pointing structure under a symmetric far-field LoS geometry. Section~\ref{sec:multiuser} tackles the general multiuser scenario using an AO algorithm based on the quadratic transform, where a safeguarded MFW method is adapted to update the RA pointing matrix. Section~\ref{sec:simulation} provides numerical results, and Section~\ref{sec:conclusion} concludes the paper.

\emph{Notations:} Boldface lowercase and uppercase letters denote vectors and matrices, respectively. $\mathbb R$, $\mathbb C$, and $\mathbb C^{M\times N}$ denote the sets of real scalars, complex scalars, and $M\times N$ complex matrices, respectively. $\bm I_N$ denotes the $N\times N$ identity matrix. For a complex scalar $x$, $|x|$ and $x^*$ denote its magnitude and complex conjugate, respectively. For a real scalar $x$, $[x]_+\triangleq\max\{0,x\}$. The Euclidean norm of a vector $\bm a$ is denoted by $\|\bm a\|$. For a matrix $\bm A$, $\|\bm A\|_F$, $\operatorname{Tr}(\bm A)$, and $[\bm A]_{i,j}$ denote its Frobenius norm, trace, and $(i,j)$-th entry, respectively. The operators $(\cdot)^T$, $(\cdot)^H$, $\operatorname{Re}\{\cdot\}$, $\operatorname{vec}(\cdot)$, and $\mathbb E\{\cdot\}$ denote transpose, conjugate transpose, real part, vectorization, and expectation, respectively. $\otimes$ denotes the Kronecker product. $\mathcal{CN}(\bm x,\bm\Sigma)$ denotes a circularly symmetric complex Gaussian distribution with mean $\bm x$ and covariance $\bm\Sigma$. Finally, $\jmath\triangleq\sqrt{-1}$ denotes the imaginary unit.

\section{System Model and Problem Formulation}
\label{sec:system_model}

\subsection{System Model}

As shown in Fig.~\ref{fig:system_model}, we consider a narrowband quasi-static half-duplex AF relaying system consisting of an $M_{\mathrm B}$-antenna BS, an $N$-antenna relay, and $K$ single-antenna users. The BS and users employ isotropic antennas, whereas the relay is equipped with a uniform planar array (UPA) of directional RAs. The sets of RAs and users are denoted by $\mathcal N\triangleq\{1,\ldots,N\}$ and $\mathcal K\triangleq\{1,\ldots,K\}$, respectively. To focus on RA-enabled two-hop transmission, we assume a relay-dominant setting in which the direct BS--user links are either blocked or sufficiently weak to be neglected. Under half-duplex operation, the relay receives the BS signal in Stage~I and amplifies and forwards it to the users in Stage~II. The boresight directions of the relay RAs are configured once per block and held fixed during both stages, thereby avoiding inter-stage reconfiguration overhead.

Without loss of generality, the relay-side UPA is assumed to lie on the $x$--$y$ plane and be centered at the origin of a three-dimensional Cartesian coordinate system. It comprises $N_x$ and $N_y$ RAs along the $x$- and $y$-axes, respectively, where $N=N_xN_y$. Let $\Delta$ denote the spacing between adjacent RAs along either axis. Using row-major indexing, the RA in the $i$-th column and $j$-th row is assigned the index $n=(j-1)N_x+i\in\mathcal N$, where $i\in\{1,\ldots,N_x\}$ and $j\in\{1,\ldots,N_y\}$. Its position vector relative to the center of the UPA can be expressed as 
\begin{equation}
	\bm w_n= \left[ \left(i-\frac{N_x+1}{2}\right)\Delta,
	\left(j-\frac{N_y+1}{2}\right)\Delta, 0\right] ^{T}.
	\label{eq:ra_position}
\end{equation}
While RA $n$ is fixed at $\bm w_n$, its boresight direction is adjustable and can be characterized by a pointing vector given by
\begin{equation}
\bm f_n=
\left[ \sin\theta_{\mathrm z,n}\cos\theta_{\mathrm a,n}, 
\sin\theta_{\mathrm z,n}\sin\theta_{\mathrm a,n}, 
\cos\theta_{\mathrm z,n}\right]^{T},
\label{eq:pointing_vector}
\end{equation}
where $\theta_{\mathrm z,n}$ denotes the zenith angle measured from the positive $z$-axis, and $\theta_{\mathrm a,n}$ denotes the azimuth angle measured in the $x$--$y$ plane from the positive $x$-axis toward the positive $y$-axis. By construction, each boresight vector has unit norm, i.e., $\left\| \bm f_n\right\| =1$. To reflect practical rotation limits, we constrain the zenith angle of each RA as 
\begin{equation}
	0\leq\theta_{\mathrm z,n}\leq\theta_{\max},\ \forall n\in\mathcal N,
	\label{eq:zenith_constraint}
\end{equation}
where $\theta_{\max}\in[0,\pi/2]$ denotes the maximum allowable zenith angle. Since $\bm f_n^{T}\bm e_3=\cos\theta_{\mathrm z,n}$, this zenith-angle
constraint is equivalently expressed as 
\begin{equation}
	\bm f_n^{T}\bm e_3\geq\cos\theta_{\max},\ \forall n\in\mathcal N,
\end{equation}
where $\bm e_3=[0,0,1]^T$ is the unit vector along the positive $z$-axis. The gain of each directional RA depends on the angular offset $\epsilon$ between the incident signal direction and its current boresight, and is modeled by the cosine pattern \cite{2016_Balanis_antenna}: 
\begin{equation}
G_{\mathrm e}(\epsilon)=
\begin{cases}
G_0\cos^{2p}(\epsilon), & 0\leq \epsilon<\pi/2,\\
0, & \text{otherwise},
\end{cases}
\label{eq:gain_pattern}
\end{equation}
where $G_0=2(2p+1)$ specifies the boresight gain, and $p > 0$ is the directivity factor controlling the main-lobe width. 

\subsubsection{BS-to-relay channel}
In Stage~I, the BS-to-relay channel is denoted by
$\bm H_1(\bm F)\in\C^{N\times M_{\mathrm B}}$, where $\bm F\triangleq[\bm f_1,\ldots,\bm f_N]$ is the pointing matrix of all RAs. For the LoS component between the $m$-th BS antenna and RA $n$, define the propagation distance and arrival direction as $r_{1,m,n}=\|\bm u_{\mathrm B,m}-\bm w_n\|$ and $\bm u_{1,m,n} =\frac{\bm u_{\mathrm B,m}-\bm w_n}{r_{1,m,n}}$, 
respectively, where $\bm u_{\mathrm B,m}$ denotes the position of the $m$-th BS antenna. Based on the RA gain pattern in
\eqref{eq:gain_pattern} and the Friis transmission equation \cite{1946_Friis_formula}, the corresponding link power gain is modeled as
\begin{equation}
	g_{1,n,m}^{(0)}(\bm f_n)
	=\beta_0G_0 r_{1,m,n}^{-2}
\left[ \bm f_n^T\bm u_{1,m,n}\right]_+^{2p},
	\label{eq:g1_los}
\end{equation}
where $\beta_0=(\lambda/4\pi)^2$ with $\lambda$ denoting the wavelength. The LoS channel coefficient is thus expressed as
\begin{equation} 
	h_{1,n,m}^{(0)}(\bm f_n)
	=\sqrt{g_{1,n,m}^{(0)}(\bm f_n)}
	e^{-\jmath\frac{2\pi}{\lambda}r_{1,m,n}}.
	\label{eq:h1_los}
\end{equation}
In addition to the LoS component, the BS-to-relay channel contains $Q_1$ single-scattering components. Let $\bm c_{1,q}$ denote the position of the $q$-th scatterer. The BS--scatterer and scatterer--RA distances are $t_{1,m,q}=\|\bm u_{\mathrm B,m}-\bm c_{1,q}\|$ and $d_{1,q,n}=\|\bm c_{1,q}-\bm w_n\|$, respectively, while $\bm c_{1,q,n}=(\bm c_{1,q}-\bm w_n)/d_{1,q,n}$ denotes the arrival direction at RA $n$. The power gain from scatterer $q$ to RA $n$ is modeled as
\begin{equation}
	g_{1,q,n}^{(\mathrm s)}(\bm f_n)
	=\beta_0G_0d_{1,q,n}^{-2}
	\left[ \bm f_n^T\bm c_{1,q,n}\right]_+^{2p}.
	\label{eq:g1_nlos}
\end{equation}
Accordingly, the channel coefficient from the $m$-th BS antenna to RA $n$ via scatterer $q$ is given by
\begin{equation}
	h_{1,n,m,q}^{(\mathrm s)}(\bm f_n)
	=\frac{\sqrt{\sigma_{1,q}
			g_{1,q,n}^{(\mathrm s)}(\bm f_n)}}{t_{1,m,q}}
	e^{-\jmath\frac{2\pi}{\lambda}(t_{1,m,q}+d_{1,q,n})
		+\jmath\chi_{1,q}},
	\label{eq:h1_nlos}
\end{equation}
where $\sigma_{1,q}$ and $\chi_{1,q}$ denote the scattering strength and phase shift of scatterer $q$, respectively. Combining the LoS component and the $Q_1$ non-line-of-sight (NLoS) components yields
\begin{equation}
	[\bm H_1(\bm F)]_{n,m}
	=h_{1,n,m}^{(0)}(\bm f_n)
	+\sum_{q=1}^{Q_1}
	h_{1,n,m,q}^{(\mathrm s)}(\bm f_n).
	\label{eq:h1}
\end{equation}

\subsubsection{Relay-to-user channel}
In Stage~II, the relay-to-user channel is denoted by $\bm H_2(\bm F)\in\C^{K\times N}$. Let $\bm u_k$ denote the position of user $k \in \mathcal K$. The LoS distance from RA $n$ to user $k$ is $r_{2,k,n}=\|\bm u_k-\bm w_n\|$, and the corresponding departure direction is $\bm u_{2,k,n}=(\bm u_k-\bm w_n)/r_{2,k,n}$. The NLoS channel of user $k$ is characterized by $Q_2$ user-specific scatterers located at $\{\bm c_{2,k,q}\}_{q=1}^{Q_2}$. For the $q$-th such scatterer, the RA--scatterer and scatterer--user distances are $d_{2,k,q,n}=\|\bm c_{2,k,q}-\bm w_n\|$ and $t_{2,k,q}=\|\bm u_k-\bm c_{2,k,q}\|$, respectively. The corresponding departure direction at RA $n$ is $\bm c_{2,k,q,n}=(\bm c_{2,k,q}-\bm w_n)/d_{2,k,q,n}$. Following the channel model in Stage~I, the $(k,n)$-th element of $\bm H_2(\bm F)$ is given by
\begin{align}
	[\bm H_2(\bm F)]_{k,n}
	= & \ \frac{\lambda\sqrt{G_0}}{4\pi r_{2,k,n}}
	\left[ \bm f_n^T\bm u_{2,k,n}\right]_+^p
	e^{-\jmath\frac{2\pi}{\lambda}r_{2,k,n}}
	\nonumber\\
	&+\sum_{q=1}^{Q_2}
	\frac{\lambda\sqrt{G_0\sigma_{2,k,q}}}
	{4\pi d_{2,k,q,n}t_{2,k,q}}
\left[ \bm f_n^T\bm c_{2,k,q,n}\right]_+^p \nonumber\\
	&\times e^{-\jmath\frac{2\pi}{\lambda}
		(d_{2,k,q,n}+t_{2,k,q})+\jmath\chi_{2,k,q}},
	\label{eq:h2}
\end{align}
where $\sigma_{2,k,q}$ and $\chi_{2,k,q}$ denote the scattering strength and phase shift of the $q$-th scatterer associated with user $k$, respectively. Let $\bm h_{2,k}^{H}(\bm F)$ denote the $k$-th row of
$\bm H_2(\bm F)$. For notational simplicity, the dependence of $\bm H_1$ and $\bm h_{2,k}$ on $\bm F$ is omitted whenever no ambiguity arises.

The BS employs space-division multiple access. Let $s_k\sim\mathcal{CN}(0,1)$ denote the information symbol intended for user $k$, where $\{s_k\}_{k=1}^{K}$ are mutually independent, and let $\bm w_{\mathrm B,k}\in\C^{M_{\mathrm B}}$ denote the corresponding
beamforming vector. Define $\bm W_{\mathrm B}\triangleq
[\bm w_{\mathrm B,1},\ldots,\bm w_{\mathrm B,K}]$ as the BS precoding matrix. In Stage~I, the transmitted signal from the BS can be expressed as $\bm x_{\mathrm B}=\sum_{k=1}^{K}\bm w_{\mathrm B,k}s_k$. The signal received at the relay is then given by
\begin{equation}
	\bm y_{\mathrm R}
	=\bm H_1\bm x_{\mathrm B}+\bm n_{\mathrm R},
	\label{eq:relay_received}
\end{equation}
where $\bm n_{\mathrm R}\sim\mathcal{CN}(\bm 0,\sigma_{\mathrm R}^2\bm I_N)$ is the receiver noise vector at the relay, and $\sigma_{\mathrm R}^2$ denotes the noise power. In Stage~II, the relay performs linear AF processing and transmits the resulting signal $\bm x_{\mathrm R}=\bm W_{\mathrm R}\bm y_{\mathrm R}$, where $\bm W_{\mathrm R}\in\C^{N\times N}$ is the relay precoding matrix. The transmit power of the relay is given by \looseness=-1
\begin{equation}
	P_{\mathrm R}
	=\mathbb E\!\left\{\bm x_{\mathrm R}^{H}\bm x_{\mathrm R}\right\}
	=\tr\!\left(\bm W_{\mathrm R}\bm R_{\mathrm{in}}
	\bm W_{\mathrm R}^{H}\right),
	\label{eq:relay_power}
\end{equation}
where $\bm R_{\mathrm{in}} \triangleq\mathbb E\!\left\{\bm y_{\mathrm R} \bm y_{\mathrm R}^{H}\right\} =\bm H_1\bm W_{\mathrm B}\bm W_{\mathrm B}^{H}\bm H_1^{H} +\sigma_{\mathrm R}^2\bm I_N$.
The received signal at user $k$ is expressed as
\begin{equation}
	y_{\mathrm U,k} =\sum_{j=1}^{K}\bm h_{2,k}^{H}\bm W_{\mathrm R}\bm H_1\bm w_{\mathrm B,j}s_j+\bm h_{2,k}^{H}\bm W_{\mathrm R}\bm n_{\mathrm R}+n_k,
	\label{eq:user_signal}
\end{equation}
where $n_k\sim\mathcal{CN}(0,\sigma_{\mathrm U,k}^2)$ is the receiver noise at user $k$, and $\sigma_{\mathrm U,k}^2$ denotes its power. Accordingly, the SINR at user $k$ can be written as
\begin{equation}
	\mathrm{SINR}_k=\frac{|A_k|^2}{B_k},
	\label{eq:sinr_compact}
\end{equation}
where $A_k\triangleq \bm h_{2,k}^{H}\bm W_{\mathrm R}\bm H_1\bm w_{\mathrm B,k}$ and $B_k\triangleq \sum_{j\ne k}\left|\bm h_{2,k}^{H}\bm W_{\mathrm R}\bm H_1\bm w_{\mathrm B,j}\right|^2+\sigma_{\mathrm R}^2\left\|\bm h_{2,k}^{H}\bm W_{\mathrm R}\right\|^2+\sigma_{\mathrm U,k}^2$. 

\subsection{Problem Formulation}
In this paper, our objective is to maximize the minimum SINR among all the users, defined as $\gamma \triangleq \min_{k\in\mathcal K}\mathrm{SINR}_k$, by jointly optimizing the BS precoding, relay precoding, and RA pointing matrices. The problem can be formulated as 
\begin{subequations}\label{prob:main}
	\begin{eqnarray}
    &\text{(P1)}:&
    \max_{\boldsymbol W_{\mathrm B},\boldsymbol W_{\mathrm R},\boldsymbol F,\gamma} \quad \gamma \label{prob:main_obj}\\
    &\text{s.t.}& \hspace{-3mm}\mathrm{SINR}_k\geq\gamma,\ \forall k \in \mathcal K,\label{prob:main_sinr}\\
    && \hspace{-3mm} \sum_{k=1}^{K}\|\bm w_{\mathrm B,k}\|^2 \leq P_{\mathrm B,\max},\label{prob:main_pb}\\
    && \hspace{-3mm} \tr\!\left(\bm W_{\mathrm R}\bm R_{\mathrm{in}}
    \bm W_{\mathrm R}^{H}\right)
    \leq P_{\mathrm R,\max},\label{prob:main_pr}\\
    && \hspace{-3mm} \bm f_n^T\bm e_3\geq\cos\theta_{\max},\ \forall n \in \mathcal N,\label{prob:main_cap}\\
    && \hspace{-3mm} \|\bm f_n\|=1,\ \forall n \in \mathcal N, \label{prob:main_norm}
	\end{eqnarray}
\end{subequations}
where $P_{\mathrm B,\max}$ and $P_{\mathrm R,\max}$ denote the maximum transmit powers of the BS and relay, respectively. 
Since $\boldsymbol W_{\mathrm B}$, $\boldsymbol W_{\mathrm R}$, and $\boldsymbol F$ are multiplicatively coupled in \eqref{prob:main_sinr} and \eqref{prob:main_pr}, and the unit-norm constraints in \eqref{prob:main_norm} render the feasible set of $\boldsymbol F$ non-convex, problem (P1) is non-convex and difficult to solve globally. To facilitate the solution of (P1), we first investigate a simplified
setting comprising a single-antenna BS and a single user, laying the groundwork for the general multiuser case. An efficient algorithm for this setting is developed in the next section, and its multiuser extension is presented in Section~\ref{sec:multiuser}. 

\section{Single-User Case With a Single-Antenna BS}
\label{sec:single_user}

In this section, we focus on the special case with a single-antenna BS and a single user, i.e., $M_{\mathrm B}=1$ and $K=1$. Since the maximum achievable SNR increases monotonically with the BS transmit power, the BS transmits with full power $P_{\mathrm B,\max}$ at optimum. For notational simplicity, define $\bm h_1(\bm F)\triangleq\bm H_1(\bm F)\in\C^N$ and $\bm h_2(\bm F)\triangleq\bm H_2^H(\bm F)\in\C^N$ as the first- and second-hop channel vectors, respectively, and omit their dependence on $\bm F$ when no ambiguity arises. Problem (P1) then reduces to
\begin{subequations}
	\begin{eqnarray}
		&\hspace{-3mm}\text{(P2)}:&
		\max_{\boldsymbol W_{\mathrm R},\boldsymbol F}\quad
		\frac{P_{\mathrm B,\max}\left|\bm h_2^H\bm W_{\mathrm R}\bm h_1\right|^2}
		{\sigma_{\mathrm R}^2\left\|\bm h_2^H\bm W_{\mathrm R}\right\|^2+\sigma_{\mathrm U}^2} \label{prob:single_obj}\\
		&\hspace{-3mm}\text{s.t.}& \hspace{-3mm}P_{\mathrm B,\max}\|\bm W_{\mathrm R}\bm h_1\|^2
		+\sigma_{\mathrm R}^2\|\bm W_{\mathrm R}\|_F^2
		\leq P_{\mathrm R,\max}, 
		\label{prob:single_pr}\\
		&& \hspace{-3mm} \eqref{prob:main_cap}, \eqref{prob:main_norm}.
	\end{eqnarray}
\end{subequations}
Problem (P2) remains non-convex because $\bm W_{\mathrm R}$ and $\bm F$ are coupled through both hops. Nevertheless, for any fixed $\bm F$, it follows from
\cite[Lemma~3]{2025_Nianzu_relay} that an optimal relay precoding matrix is given by
\begin{equation}
	\bm W_{\mathrm R}^{\star}(\bm F)
	=\beta(\bm F)\bm h_2\bm h_1^H,
	\label{eq:single_opt_W}
\end{equation}
where the power-normalization factor $\beta(\bm F)>0$ follows from the active relay power constraint \eqref{prob:single_pr} as 
\begin{equation}
	\beta(\bm F)=
	\sqrt{\frac{P_{\mathrm R,\max}}
		{\|\bm h_2\|^2
			\left(P_{\mathrm B,\max}\|\bm h_1\|^4
			+\sigma_{\mathrm R}^2\|\bm h_1\|^2\right)}}.
	\label{eq:single_beta}
\end{equation}
For compactness, define
$y_1(\bm F)\triangleq\|\bm h_1\|^2$ and
$y_2(\bm F)\triangleq\|\bm h_2\|^2$. Substituting \eqref{eq:single_opt_W} into the objective function of (P2) yields the maximum SNR for a given $\bm F$, denoted by $\gamma_{\mathrm{su}}^{\star}(\bm F)$, as
\begin{equation}
	\gamma_{\mathrm{su}}^{\star}(\bm F)=
	\frac{P_{\mathrm B,\max}P_{\mathrm R,\max}
		y_1(\bm F)y_2(\bm F)}
	{P_{\mathrm B,\max}\sigma_{\mathrm U}^2y_1(\bm F)
		+P_{\mathrm R,\max}\sigma_{\mathrm R}^2y_2(\bm F)
		+\sigma_{\mathrm R}^2\sigma_{\mathrm U}^2}.
	\label{eq:single_snr}
\end{equation} 
Introducing the normalized power budgets $\Gamma_{\mathrm B}\triangleq P_{\mathrm B,\max}/\sigma_{\mathrm R}^2$ and $\Gamma_{\mathrm R}\triangleq P_{\mathrm R,\max}/\sigma_{\mathrm U}^2$, equality \eqref{eq:single_snr} can be compactly rewritten as
$\gamma_{\mathrm{su}}^{\star}(\bm F)=\frac{\Gamma_{\mathrm B}\Gamma_{\mathrm R}y_1(\bm F)y_2(\bm F)}{\bar D(\bm F)}$, where $\bar D(\bm F)\triangleq \Gamma_{\mathrm B}y_1(\bm F)+\Gamma_{\mathrm R}y_2(\bm F)+1$. Since the logarithm is strictly increasing, maximizing
$\gamma_{\mathrm{su}}^{\star}(\bm F)$ is equivalent to maximizing its logarithm. After dropping the constant
$\ln(\Gamma_{\mathrm B}\Gamma_{\mathrm R})$, problem (P2) is reduced to
\begin{subequations}
	\label{prob:single_reduced}
	\begin{align}
		\max_{\boldsymbol F}\quad
		&\ln y_1(\bm F)+\ln y_2(\bm F)
		-\ln\bar D(\bm F)
		\label{prob:single_reduced_obj}\\
		\text{s.t.}\quad
		&\eqref{prob:main_cap},\ \eqref{prob:main_norm}.
	\end{align}
\end{subequations}
Thus, only $\bm F$ needs to be optimized, after which
$\bm W_{\mathrm R}^{\star}$ follows from
\eqref{eq:single_opt_W}. Despite this simplification, problem \eqref{prob:single_reduced} remains difficult to solve because its objective function is non-concave in $\bm F$ and the unit-norm constraints in \eqref{prob:main_norm} yield a non-convex feasible region. 

\subsection{Proposed Algorithm for Problem \eqref{prob:single_reduced}}

The constraints \eqref{prob:main_cap} and \eqref{prob:main_norm} restrict each pointing vector to the spherical cap
\begin{equation}
	\mathcal C_{\mathrm{cap}}
	\triangleq
	\left\{\bm x\in\R^3:
	\|\bm x\|=1,\ 
	\bm x^T\bm e_3\geq\cos\theta_{\max}\right\}.
	\label{eq:spherical_cap}
\end{equation}
Hence, the feasible set of $\bm F$ is the Cartesian product $\mathcal C_{\mathrm{cap}}^N \triangleq  \prod_{n=1}^N\mathcal C_{\mathrm{cap}}$. Although non-convex, its product structure makes the tangent-space linear oracle separable across the RAs and permits a closed-form solution. Motivated by this structure, we solve problem \eqref{prob:single_reduced} using an MFW method that combines an explicit expression for the objective gradient with a closed-form linear oracle and a feasibility-preserving retraction. 

At iteration $r$, given the current RA pointing matrix $\bm F^r$, we compute the blockwise Euclidean gradients of the objective function in \eqref{prob:single_reduced_obj} with respect to $\bm f_n$, $n\in\mathcal N$. Denoting this objective function by $J(\bm F)\triangleq\ln y_1(\bm F)+\ln y_2(\bm F)-\ln\bar D(\bm F)$, the chain rule gives
\begin{align}
	\bm g_n^r
	\triangleq{} &
	\left.\nabla_{\boldsymbol f_n}J(\bm F)\right|_{\boldsymbol F=\boldsymbol F^r} = \left(\frac{1}{y_1(\bm F^r)}
	-\frac{\Gamma_{\mathrm B}}{\bar D(\bm F^r)}
	\right)\bm g_{y_1,n}^r \nonumber\\ 
	& + \left(\frac{1}{y_2(\bm F^r)}
	-\frac{\Gamma_{\mathrm R}}{\bar D(\bm F^r)}
	\right)\bm g_{y_2,n}^r,
	\label{eq:single_objective_gradient}
\end{align}
where $\bm g_{y_i,n}^r \triangleq \left. \nabla_{\boldsymbol f_n}y_i(\bm F)\right|_{\boldsymbol F=\boldsymbol F^r}$, $i\in\{1,2\}$. To evaluate $\bm g_{y_i,n}^r$, let $h_{i,n}(\bm f_n)$ denote the $n$-th element of $\bm h_i(\bm F)$. Since $h_{i,j}(\bm f_j)$ is independent of $\bm f_n$ for $j\neq n$, only the $n$-th term in $y_i(\bm F)=\sum_{j=1}^{N}|h_{i,j}(\bm f_j)|^2$ contributes to $\bm g_{y_i,n}^r$, yielding 
\begin{align}
	\bm g_{y_i,n}^r =
	2\operatorname{Re}\!\left\{
	h_{i,n}^{*}(\bm f_n^r)
	\left.
	\nabla_{\boldsymbol f_n}h_{i,n}(\bm f_n)
	\right|_{\boldsymbol f_n=\boldsymbol f_n^r}
	\right\}.
	\label{eq:power_gradient}
\end{align}
To obtain the channel-coefficient gradient in \eqref{eq:power_gradient}, we collect all terms independent of $\bm f_n$ and write $h_{i,n}(\bm f_n)$ as
\begin{equation}
	h_{i,n}(\bm f_n)
	=
	\sum_{\ell=0}^{Q_i}
	a_{i,\ell,n}
	\left[ \bm f_n^T\bm v_{i,\ell,n}\right]_+^p,
	\label{eq:generic_channel}
\end{equation}
where $\ell=0$ denotes the LoS path, whereas
$\ell\in\{1,\ldots,Q_i\}$ indexes the NLoS paths. Moreover,
$\bm v_{1,\ell,n}$ and $\bm v_{2,\ell,n}$ denote the arrival and departure directions of path $\ell$ at RA $n$, respectively. The coefficient $a_{i,\ell,n}$ collects the orientation-independent propagation attenuation and phase as well as, for $\ell\geq1$, the scattering coefficient. Therefore, we have
\begin{align}
	&\left.
	\nabla_{\boldsymbol f_n}h_{i,n}(\bm f_n)
	\right|_{\boldsymbol f_n=\boldsymbol f_n^r} \nonumber\\
	&=
	p\sum_{\substack{\ell=0\\
			(\boldsymbol f_n^r)^T\boldsymbol v_{i,\ell,n}>0}}^{Q_i}
	a_{i,\ell,n}
	\left((\bm f_n^r)^T\bm v_{i,\ell,n}\right)^{p-1}
	\bm v_{i,\ell,n}.
	\label{eq:fundamental_gradient}
\end{align}
At $(\bm f_n^r)^T\bm v_{i,\ell,n}=0$, where the classical derivative may not exist, the corresponding term in \eqref{eq:fundamental_gradient} is set to zero for numerical implementation. 
Substituting \eqref{eq:fundamental_gradient} into \eqref{eq:power_gradient} yields $\bm g_{y_i,n}^r$, and hence $\bm g_n^r$ follows from \eqref{eq:single_objective_gradient}. 

Since the unit-norm constraint permits only first-order variations orthogonal to $\bm f_n^r$, the Euclidean gradient must be projected onto the tangent space of the unit sphere. As $\|\bm f_n^r\|=1$, the corresponding orthogonal projector is $\bm I_3-\bm f_n^r(\bm f_n^r)^T$. Hence, the projected gradient, denoted by $\bm p_n^r$, is given by $\bm p_n^r = \left(\bm I_3-\bm f_n^r(\bm f_n^r)^T\right)\bm g_n^r$. Using the projected gradients, the MFW linear oracle selects a feasible reference pointing vector $\bm s_n^r\in\mathcal C_{\mathrm{cap}}$ for each RA to maximize the local first-order ascent. Since $\mathcal C_{\mathrm{cap}}^N$ is a Cartesian product, these reference vectors can be determined independently as
\begin{equation}
	\bm s_n^r
	\in
	\arg\max_{\widetilde{\boldsymbol f}_n\in\mathcal C_{\mathrm{cap}}}
	(\bm p_n^r)^T\widetilde{\bm f}_n.
	\label{eq:linear_oracle_problem}
\end{equation}
When $\bm p_n^r=\bm 0$, every point in $\mathcal C_{\mathrm{cap}}$ is optimal for \eqref{eq:linear_oracle_problem}. Selecting $\bm s_n^r=\bm f_n^r$ yields a zero search direction. For $\bm p_n^r\neq\bm 0$, let $\widehat{\bm p}_n^r\triangleq\bm p_n^r/\|\bm p_n^r\|$ denote its unit direction and $\widehat{\bm p}_{n,\mathrm{xy}}^r\triangleq(\bm I_3-\bm e_3\bm e_3^T)\widehat{\bm p}_n^r$ its projection onto the $x$--$y$ plane. The solution to \eqref{eq:linear_oracle_problem} is then given by 
\begin{equation}
	\bm s_n^r
	=
	\begin{cases}
		\widehat{\bm p}_n^r,
		&(\widehat{\bm p}_n^r)^T\bm e_3\geq c_{\max},\\[1mm]
		\displaystyle
		\left[
		s_{\max}
		\frac{(\widehat{\bm p}_{n,\mathrm{xy}}^r)^T}
		{\left\| \widehat{\bm p}_{n,\mathrm{xy}}^r\right\|},
		c_{\max}
		\right]^T,
		&\substack{(\widehat{\bm p}_n^r)^T\bm e_3<c_{\max},\\
			\left\| \widehat{\bm p}_{n,\mathrm{xy}}^r\right\| >0},\\[2mm]
		\left[s_{\max}\bm u^T,c_{\max}\right]^T,
		&\substack{(\widehat{\bm p}_n^r)^T\bm e_3<c_{\max},\\
		\left\| \widehat{\bm p}_{n,\mathrm{xy}}^r\right\|=0}, 
	\end{cases}
	\label{eq:closed_form_oracle}
\end{equation}
where $c_{\max}\triangleq\cos\theta_{\max}$, $s_{\max}\triangleq\sin\theta_{\max}$, and $\bm u\in\R^2$ is any unit vector. The third case corresponds to $\widehat{\bm p}_n^r=-\bm e_3$, for which every point on the boundary of $\mathcal C_{\mathrm{cap}}$ is optimal. 

The MFW search direction from the current pointing vector $\bm f_n^r$ toward the reference point $\bm s_n^r$ is defined as
\begin{equation}
	\bm d_n^r
	\triangleq
	\bm s_n^r-\bm f_n^r.
	\label{eq:mfw_direction}
\end{equation} 
For a trial step size $\rho\in(0,1]$ satisfying
$\bm f_n^r+\rho\bm d_n^r\neq\bm 0$ for all $n\in\mathcal N$, the candidate pointing vectors are obtained by moving from $\bm f_n^r$ along $\bm d_n^r$ and retracting the resulting vectors onto the unit sphere as
\begin{equation}
	\bm f_n^{r+1}(\rho)
	=
	\frac{\bm f_n^r+\rho\bm d_n^r}
	{\left\|\bm f_n^r+\rho\bm d_n^r\right\|},
	\quad \forall n\in\mathcal N.
	\label{eq:retraction}
\end{equation}
Since $\bm f_n^r+\rho\bm d_n^r=(1-\rho)\bm f_n^r+\rho\bm s_n^r$, the numerator in \eqref{eq:retraction} is a convex combination of two points in $\mathcal C_{\mathrm{cap}}$. Hence, its component along $\bm e_3$ is at least $c_{\max}$ and its norm is no greater than one. Since $c_{\max}\geq0$, the normalization preserves the tilt constraint while enforcing the unit-norm constraint, ensuring that $\bm f_n^{r+1}(\rho)\in\mathcal C_{\mathrm{cap}}$. Collecting the candidate pointing vectors gives $\bm F^{r+1}(\rho) \triangleq \left[ \bm f_1^{r+1}(\rho),\ldots,\bm f_N^{r+1}(\rho)\right]$.

To determine how far to move along the MFW directions, define the aggregate first-order gain as $\Delta^r \triangleq \sum_{n=1}^{N}(\bm p_n^r)^T\bm d_n^r$, 
which is referred to as the MFW gap. The step size is then selected by Armijo backtracking to ensure a sufficient increase in the objective function. Starting from $\rho=1$, the trial step size is successively updated via $\rho\leftarrow\tau\rho$ with $\tau\in(0,1)$ until
\begin{equation}
	J\bigl(\bm F^{r+1}(\rho)\bigr)
	\geq
	J(\bm F^r)
	+c_{\mathrm A}\rho\Delta^r
	\label{eq:armijo}
\end{equation}
is satisfied or a prescribed maximum number $I_{\mathrm A}$ of trial step sizes has been examined, where
$c_{\mathrm A}\in(0,1)$ is the Armijo parameter. If an acceptable step size $\rho^r$ is found, the pointing matrix is updated as $\bm F^{r+1}=\bm F^{r+1}(\rho^r)$; otherwise, the current pointing matrix is retained, i.e., $\bm F^{r+1}=\bm F^r$.

Since $\bm f_n^r$ is feasible for the linear oracle in
\eqref{eq:linear_oracle_problem}, we have $(\bm p_n^r)^T(\bm s_n^r-\bm f_n^r)\geq0$ for every $n\in\mathcal N$, and hence $\Delta^r\geq0$. Therefore, every accepted update satisfies $J(\bm F^{r+1})\geq J(\bm F^r)$ according to
\eqref{eq:armijo}, while an unsuccessful line search leaves the objective unchanged. Consequently, the proposed MFW algorithm generates a non-decreasing objective sequence. This monotonicity, together with the upper boundedness of $J(\bm F)$ over the compact feasible set $\mathcal C_{\mathrm{cap}}^N$, guarantees that $\{J(\bm F^r)\}$ converges. Regarding complexity, for fixed numbers of scattering paths, computing all block gradients and closed-form linear oracles requires $\mathcal O(N)$ operations per iteration. The subsequent Armijo line search reevaluates the objective function at each trial, also with complexity $\mathcal O(N)$. Hence, for $I_{\mathrm{MFW}}$ iterations and at most $I_{\mathrm A}$ Armijo trials per iteration, the overall complexity is given by $\mathcal O(I_{\mathrm{MFW}}I_{\mathrm A}N)$.

\subsection{Antenna Pointing Structure Under Symmetric Far-Field LoS Geometry}

To gain analytical insight into the RA pointing structure, we further consider a symmetric far-field LoS geometry in which the BS and user are equidistant from the relay. The first- and second-hop propagation directions lie in the same plane and form signed angles $+\varphi$ and $-\varphi$ with the positive $z$-axis, respectively, where $\varphi\in(0,\pi/2)$. Thus, the positive $z$-axis bisects the two directions. Projecting any feasible pointing vector onto the plane containing the two propagation directions and renormalizing the result preserves feasibility without decreasing either hop gain, so restricting the pointing vectors to that plane incurs no loss of optimality. Let $\theta_n\in[-\theta_{\max},\theta_{\max}]$ denote the signed pointing angle of RA $n$ from the positive $z$-axis. When $p>0$, $N$ is even, $\theta_{\max}\geq\varphi$, and $\Gamma_{\mathrm B}=\Gamma_{\mathrm R}$, the globally optimal RA pointing structure is characterized as follows.

\begin{figure}[!t]
	\vspace{-2mm}
	\centering
	\includegraphics[scale=0.66]{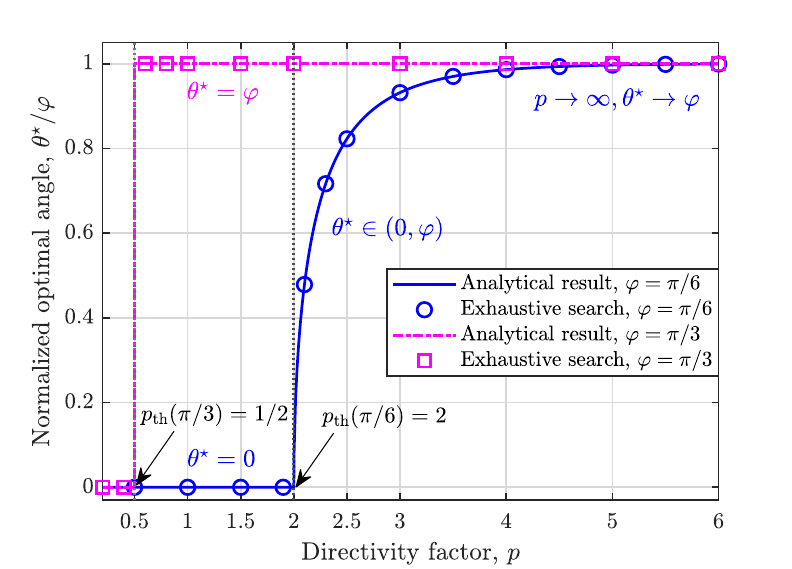}
	\caption{Normalized optimal RA pointing angle versus directivity factor.}
	\label{fig:prop_verification}
	\vspace{-1mm}
\end{figure}
 
\begin{proposition}
	\label{prop:splitting}
	Define the critical directivity factor as
	\begin{equation}
		p_{\mathrm{th}}(\varphi)
		\triangleq
		\begin{cases}
			\displaystyle\frac{1}{2\sin^2\varphi},
			& 0<\varphi<\pi/4,\\[2mm]
			\displaystyle\frac{\ln 2}{-2\ln(\cos\varphi)},
			& \pi/4\leq\varphi< \pi/2.
		\end{cases}
		\label{eq:splitting_threshold}
	\end{equation}
   A globally optimal RA pointing configuration can be constructed from a non-negative angle $\theta^\star$ under the following two regimes.

   (i) If $0<p\leq p_{\mathrm{th}}(\varphi)$, then $\theta^\star=0$, and setting $\bm f_n^\star=\bm e_3$ for all $n\in\mathcal N$ yields a globally optimal solution to problem \eqref{prob:single_reduced}.
   
   (ii) If $p>p_{\mathrm{th}}(\varphi)$, then $\theta^\star>0$. For $0<\varphi<\pi/4$, $\theta^\star$ is the unique solution in $(0,\varphi)$ to
   \begin{equation}
   	\frac{\sin(\varphi-\theta^\star)}
   	{\sin(\varphi+\theta^\star)}
   	=
   	\left[
   	\frac{\cos(\varphi+\theta^\star)}
   	{\cos(\varphi-\theta^\star)}
   	\right]^{2p-1}.
   	\label{eq:optimal_splitting_angle}
   \end{equation}
   This solution can be obtained via one-dimensional root finding. As $p$ increases above $p_{\mathrm{th}}(\varphi)$, $\theta^\star$ increases continuously from zero and approaches $\varphi$ as $p\to\infty$. For $\pi/4\leq\varphi<\pi/2$, we have $\theta^\star=\varphi$. In either case, a globally optimal solution to problem \eqref{prob:single_reduced} is obtained by pointing $N/2$ RAs at $+\theta^\star$ and the remaining $N/2$ RAs at $-\theta^\star$. 
\end{proposition}

\begin{proof}
    See the Appendix.
\end{proof}	

\begin{remark} \label{rem:direct_single}
	\rm When the direct BS--user link is non-negligible, applying maximal-ratio combining to the signals received over the direct and relay links yields an end-to-end SNR of $\gamma_{\mathrm d}+\gamma_{\mathrm{su}}^\star(\bm F)$, where $\gamma_{\mathrm d}$ denotes the direct-link SNR. Since $\gamma_{\mathrm d}$ is independent of $\bm W_{\mathrm R}$ and $\bm F$, the proposed single-user MFW algorithm and the pointing characterization in Proposition~\ref{prop:splitting} remain directly applicable.
\end{remark} 

Proposition~\ref{prop:splitting} reveals how antenna directivity governs the trade-off between balancing the two hops and aligning with each hop individually. A broad antenna pattern allows common pointing to provide sufficient gain over both hops, whereas a narrower pattern makes such simultaneous coverage inefficient and thereby favors antenna splitting.

Fig.~\ref{fig:prop_verification} illustrates the characterization in Proposition~\ref{prop:splitting} for $\varphi=\pi/6$ and $\varphi=\pi/3$, which represent the two angular regimes considered therein. To verify this characterization numerically, we also obtain the optimal pointing angles through exhaustive search. The analytical and exhaustive-search results agree closely, corroborating the derived thresholds and illustrating the transition from common pointing to antenna splitting as $p$ exceeds the $\varphi$-dependent threshold $p_{\mathrm{th}}(\varphi)$. 

\section{General Multiuser Case}
\label{sec:multiuser}

Building on the MFW procedure introduced for the single-user case, we now address problem (P1) under the general system setting with an $M_{\mathrm B}$-antenna BS and $K$ users. In this setting, inter-user interference makes the SINR of each user jointly dependent on $\bm W_{\mathrm B}$, $\bm W_{\mathrm R}$, and $\bm F$, and the relay precoding matrix no longer admits the closed-form solution available in the single-user case. Although AO is a natural approach to handling the coupling among these variables, applying it directly to (P1) would leave fractional SINR constraints in the resulting subproblems, complicating the efficient optimization of each variable block. We therefore apply the quadratic transform \cite{2018_Kaiming_FP} to recast (P1) into an equivalent auxiliary-variable formulation amenable to AO. Within the AO iterations, the RA pointing matrix $\bm F$ is updated 
using a safeguarded extension of the MFW method. 

Specifically, recall that $\mathrm{SINR}_k=|A_k|^2/B_k$, where $B_k>0$. For an auxiliary variable $\alpha_k\in\C$, we have
\begin{equation}
	\frac{|A_k|^2}{B_k}
	-B_k\left|\alpha_k-\frac{A_k}{B_k}\right|^2
	=
	2\operatorname{Re}\{\alpha_k^*A_k\}
	-|\alpha_k|^2B_k.
\end{equation}
The subtracted term is non-negative and becomes zero at $\alpha_k=A_k/B_k$. Maximizing the above expression over $\alpha_k$ therefore gives
\begin{equation}
	\frac{|A_k|^2}{B_k}
	=
	\max_{\alpha_k\in\C}
	\left(
	2\operatorname{Re}\{\alpha_k^*A_k\}
	-|\alpha_k|^2B_k
	\right),
	\label{eq:qt_identity}
\end{equation}
with the optimal auxiliary variable given by $\alpha_k^\star=A_k/B_k$. Let $\bm\alpha\triangleq[\alpha_1,\ldots,\alpha_K]^T\in\C^K$ collect all auxiliary variables, and define $\hat g_k(\bm W_{\mathrm B},\bm W_{\mathrm R},\bm F,\alpha_k)\triangleq2\operatorname{Re}\{\alpha_k^*A_k\}-|\alpha_k|^2B_k$. For notational simplicity, the arguments of $\hat g_k$ are omitted hereafter whenever no ambiguity arises. Problem (P1) can then be equivalently reformulated as 
\begin{subequations}\label{prob:fp}
	\begin{eqnarray}
		&\text{(P3)}:&
		\max_{\substack{\boldsymbol W_{\mathrm B},\boldsymbol W_{\mathrm R},\boldsymbol F,\\
				\boldsymbol\alpha,\gamma}}
		\hspace{2mm} \gamma
		\label{prob:fp_obj}\\
		&\text{s.t.}&  \hspace{-3mm} \hat g_k \geq\gamma,\  \forall k\in\mathcal K,\label{prob:fp_sinr}\\
		&& \hspace{-3mm} \eqref{prob:main_pb}-\eqref{prob:main_norm}.
	\end{eqnarray}
\end{subequations}
To solve problem (P3), we employ AO over the three design-variable blocks $\bm F$, $\bm W_{\mathrm B}$, and $\bm W_{\mathrm R}$. Before each block update, $\bm\alpha$ is refreshed as $\alpha_k\leftarrow A_k/B_k$, $\forall k\in\mathcal K$, using the latest values of the design variables. This refresh restores the tightness of the quadratic-transform representation at the current iterate, ensuring that each accepted block update does not decrease the original minimum-SINR objective. The three block updates are detailed below. 

\subsection{RA Pointing Matrix Update}

For fixed $\bm\alpha$, $\bm W_{\mathrm B}$, and $\bm W_{\mathrm R}$, the RA pointing subproblem involves the same spherical-cap constraints as in the single-user case, motivating an extension of the MFW method developed therein to the present setting. This extension, however, requires addressing two additional issues. First, the $\bm F$-dependent relay power constraint necessitates a feasibility safeguard. Second, constructing the search direction requires a differentiable scalar objective function in $\bm F$. To address the second issue, we eliminate $\gamma$ from the transformed SINR constraints and equivalently write the RA pointing subproblem in the following max-min form:
\begin{align}\label{prob:F_nonsmooth}
	\max_{\bm F}\hspace{2mm} \min_{k\in\mathcal K} \hspace{1mm} \hat g_k(\bm F) \hspace{7mm}
	\text{s.t.}\hspace{3mm} \eqref{prob:main_pr}-\eqref{prob:main_norm}.
\end{align}
The objective $g_{\min}(\bm F)\triangleq\min_{k\in\mathcal K}\hat g_k(\bm F)$ in \eqref{prob:F_nonsmooth} is non-smooth. To obtain the gradient required by the MFW update, we employ a smooth log-sum-exp approximation. Specifically, define
$M(\bm F)\triangleq\max_{k\in\mathcal K}\{-\hat g_k(\bm F)\}=-g_{\min}(\bm F)$.
For any $\mu>0$, the definition of $M(\bm F)$ gives
$e^{-\mu \hat g_k(\boldsymbol F)}\leq e^{\mu M(\boldsymbol F)}$ for every $k\in\mathcal K$, with equality for at least one $k$. Consequently, 
$e^{\mu M(\boldsymbol F)}\leq\sum_{k=1}^{K}e^{-\mu \hat g_k(\boldsymbol F)} \leq K e^{\mu M(\boldsymbol F)}$.
Taking the logarithm and dividing by $\mu$ yields
\begin{equation}
	M(\bm F)
	\leq
	\frac{1}{\mu}\ln\!\left(
	\sum_{k=1}^{K}e^{-\mu \hat g_k\left( \boldsymbol F\right)}
	\right)
	\leq
	M(\bm F)+\frac{\ln(K)}{\mu}.
	\label{eq:log_sum_exp_bound}
\end{equation}
Since $M(\bm F)=-g_{\min}(\bm F)$, reversing the signs in \eqref{eq:log_sum_exp_bound} yields
\begin{align}
	g_{\min}(\bm F)-\frac{\ln(K)}{\mu}
	\leq 
	-\frac{1}{\mu}\ln\!\left(
	\sum_{k=1}^{K}e^{-\mu \hat g_k\left( \boldsymbol F\right) }
	\right) 
	\leq 
	g_{\min}(\bm F).
	\label{eq:smooth_min_bound}
\end{align}
Thus, $G_\mu(\bm F) \triangleq 	-\frac{1}{\mu}\ln\!\left(
\sum_{k=1}^{K}e^{-\mu \hat g_k\left( \boldsymbol F\right)}\right)$ is a smooth lower bound on $g_{\min}(\bm F)$, with an approximation error no greater than $\ln(K)/\mu$. Based on this smooth approximation, we next derive a safeguarded MFW update for $\bm F$. The gradient of $G_\mu(\bm F)$ is used to construct the search direction. A candidate update is accepted only if it does not decrease $g_{\min}(\bm F)$ and satisfies the relay power constraint. 

In the derivation below, we omit the outer AO-iteration index for notational simplicity. At MFW iteration $r$, given the current RA pointing matrix $\bm F^r$, the blockwise Euclidean gradient of $G_\mu(\bm F)$ with respect to $\bm f_n$ is given by
\begin{equation}
	\bm g_{\mu,n}^r
	\triangleq
	\left.
	\nabla_{\boldsymbol f_n}G_\mu(\bm F)
	\right|_{\boldsymbol F=\boldsymbol F^r}
	=
	\sum_{k=1}^{K}\frac{e^{-\mu \hat g_k(\boldsymbol F^r)}} {\sum_{j=1}^{K}e^{-\mu \hat g_j(\boldsymbol F^r)}}\bm g_{k,n}^r,
	\label{eq:lse_gradient}
\end{equation}
where $\bm g_{k,n}^r \triangleq\left.\nabla_{\boldsymbol f_n}\hat g_k(\bm F) \right|_{\boldsymbol F=\boldsymbol F^r}$.  For fixed $\alpha_k$, the chain rule gives
\begin{equation}
	\bm g_{k,n}^r
	=
	2\operatorname{Re}\!\left\{
	\alpha_k^*\bm a_{k,n}^r
	\right\}
	-
	|\alpha_k|^2\bm b_{k,n}^r,
	\label{eq:gradient_gk_basic}
\end{equation}
where $\bm a_{k,n}^r \triangleq \left. \nabla_{\boldsymbol f_n}A_k(\bm F) \right|_{\boldsymbol F=\boldsymbol F^r}$ and $\bm b_{k,n}^r \triangleq \left. \nabla_{\boldsymbol f_n}B_k(\bm F) \right|_{\boldsymbol F=\boldsymbol F^r}$. To evaluate $\bm a_{k,n}^r$, observe that $\bm f_n$ affects $A_k$ only through the $n$-th row of $\bm H_1$ and the $n$-th element of $\bm h_{2,k}^{H}$. Let $\bm H_1^r\triangleq\bm H_1(\bm F^r)$ and $\bm h_{2,k}^r\triangleq\bm h_{2,k}(\bm F^r)$ denote the corresponding channel values at the current iterate. Denote their block derivatives by $\bm d_{1,n,m}^r\triangleq\left.\nabla_{\boldsymbol f_n}[\bm H_1(\bm F)]_{n,m}\right|_{\boldsymbol F=\boldsymbol F^r}$ and $\bm d_{2,k,n}^r\triangleq\left.\nabla_{\boldsymbol f_n}[\bm h_{2,k}^{H}(\bm F)]_n\right|_{\boldsymbol F=\boldsymbol F^r}$. Applying the product rule yields
\begin{align}
	\bm a_{k,n}^r
	={}&
	[\bm W_{\mathrm R}\bm H_1^r\bm w_{\mathrm B,k}]_n
	\bm d_{2,k,n}^r
	\nonumber\\
	&+
	[(\bm h_{2,k}^r)^H\bm W_{\mathrm R}]_n
	\sum_{m=1}^{M_{\mathrm B}}
	w_{\mathrm B,k,m}\bm d_{1,n,m}^r.
	\label{eq:grad_Ak}
\end{align}
To evaluate $\bm b_{k,n}^r$, we separately consider the inter-user interference and forwarded relay-noise terms in $B_k$. For $j\neq k$, define the effective interference coefficient as $S_{k,j}(\bm F)\triangleq\bm h_{2,k}^{H}(\bm F)\bm W_{\mathrm R}\bm H_1(\bm F)\bm w_{\mathrm B,j}$, and let $S_{k,j}^r\triangleq S_{k,j}(\bm F^r)$ and $\bm z_{k,j,n}^r\triangleq\left.\nabla_{\boldsymbol f_n}S_{k,j}(\bm F)\right|_{\boldsymbol F=\boldsymbol F^r}$. Since $S_{k,j}(\bm F)$ differs from $A_k(\bm F)$ only in the BS beamforming vector, $\bm z_{k,j,n}^r$ follows from \eqref{eq:grad_Ak} by replacing $\bm w_{\mathrm B,k}$ with $\bm w_{\mathrm B,j}$. The user-noise term is independent of $\bm F$. Differentiating the remaining terms in $B_k$ therefore gives
\begin{align}
	\bm b_{k,n}^r
	={}&
	2\sum_{j\neq k}
	\operatorname{Re}\!\left\{
	(S_{k,j}^r)^*\bm z_{k,j,n}^r
	\right\}
	\nonumber\\
	&+
	2\sigma_{\mathrm R}^2
	\operatorname{Re}\!\left\{
	[\bm W_{\mathrm R}\bm W_{\mathrm R}^{H}
	\bm h_{2,k}^r]_n
	\bm d_{2,k,n}^r
	\right\}.
	\label{eq:grad_Bk}
\end{align} 
Despite the additional indices $m$ and $k$, the multiuser channel coefficients retain the same pathwise dependence on $\bm f_n$ as in \eqref{eq:generic_channel}. Thus, $\bm d_{1,n,m}^r$ and $\bm d_{2,k,n}^r$ are obtained by applying \eqref{eq:fundamental_gradient} pathwise, and their explicit expressions are omitted for brevity. With these derivatives, \eqref{eq:grad_Ak} and \eqref{eq:grad_Bk} determine $\bm g_{k,n}^r$ through \eqref{eq:gradient_gk_basic}, which in turn yields $\bm g_{\mu,n}^r$ from \eqref{eq:lse_gradient}. 

With $\bm g_{\mu,n}^r$ available, we follow the single-user MFW construction and project it onto the tangent space at $\bm f_n^r$ as $\bm p_{\mu,n}^r\triangleq\left(\bm I_3-\bm f_n^r(\bm f_n^r)^T\right)\bm g_{\mu,n}^r$. Applying the closed-form oracle in \eqref{eq:closed_form_oracle} to $\bm p_{\mu,n}^r$ yields $\bm s_{\mu,n}^r$. The resulting MFW direction is $\bm d_{\mu,n}^r\triangleq\bm s_{\mu,n}^r-\bm f_n^r$. The associated MFW gap is defined as $\Delta_\mu^r\triangleq\sum_{n=1}^{N}(\bm p_{\mu,n}^r)^T\bm d_{\mu,n}^r$. Since $\bm f_n^r$ is feasible for the $n$-th linear oracle, $\Delta_\mu^r$ is non-negative and serves as a first-order ascent measure. For a trial step size $\rho\in(0,1]$ satisfying
$\bm f_n^r+\rho\bm d_{\mu,n}^r\neq\bm 0$ for all $n\in\mathcal N$, we form the candidate matrix $\bm F^{r+1}(\rho)$ by normalizing each $\bm f_n^r+\rho\bm d_{\mu,n}^r$ in the same manner as in \eqref{eq:retraction}. By the same argument as in the single-user case, $\bm F^{r+1}(\rho)$ satisfies the spherical-cap constraints. However, the retraction alone does not enforce the $\bm F$-dependent relay power constraint, and an increase in $G_\mu(\bm F)$ does not necessarily imply an increase in $g_{\min}(\bm F)$. The Armijo line search is therefore supplemented with two safeguards. Specifically, starting from $\rho=1$, the trial step size is successively updated via $\rho\leftarrow\tau\rho$ with $\tau\in(0,1)$, until the following conditions are satisfied or a prescribed maximum number $I_{\mathrm A}$ of trial step sizes is reached:   
\begin{align}
	G_\mu(\bm F^{r+1}(\rho))
	&\geq
	G_\mu(\bm F^r)
	+c_{\mathrm A}\rho\Delta_\mu^r,
	\label{eq:multi_armijo}\\
	\tr\!\left(
	\bm W_{\mathrm R}
	\bm R_{\mathrm{in}}(\bm F^{r+1}(\rho))
	\bm W_{\mathrm R}^{H}
	\right)
	&\leq P_{\mathrm R,\max},
	\label{eq:multi_power_safeguard}\\
	g_{\min}(\bm F^{r+1}(\rho))
	&\geq g_{\min}(\bm F^r).
	\label{eq:multi_merit_safeguard}
\end{align}
The first condition ensures a sufficient increase in the smoothed objective, whereas the other two preserve relay power feasibility and ensure that $g_{\min}(\bm F)$ is non-decreasing, respectively. To relate the third condition to the original objective, recall that the quadratic transform satisfies $\hat g_k(\bm F,\alpha_k)\leq\mathrm{SINR}_k(\bm F)$ for any $\alpha_k$, with equality at $\bm F^r$ after refreshing $\alpha_k=A_k(\bm F^r)/B_k(\bm F^r)$. Hence, \eqref{eq:multi_merit_safeguard} ensures that the RA update does not decrease the original minimum SINR. Upon finding a step size $\rho^r$ satisfying the above conditions, we update the pointing matrix as $\bm F^{r+1}=\bm F^{r+1}(\rho^r)$. If none of the $I_{\mathrm A}$ trial step sizes satisfies all three conditions, we retain the current pointing matrix by setting $\bm F^{r+1}=\bm F^r$.

\subsection{BS Precoding Matrix Update}

For fixed $\bm\alpha$, $\bm F$, and $\bm W_{\mathrm R}$, define the effective cascaded channel for user $k$ as $\bar{\bm h}_k^{H}\triangleq\bm h_{2,k}^{H}(\bm F)\bm W_{\mathrm R}\bm H_1(\bm F)$, which collects the two-hop channel and relay precoding terms independent of $\bm W_{\mathrm B}$. Then, $\hat g_k$ takes the following equivalent form: \looseness=-1
\begin{align}
	\hat g_k(\bm W_{\mathrm B})
	={}&2\operatorname{Re}\!\left\{
	\alpha_k^*\bar{\bm h}_k^H\bm w_{\mathrm B,k}
	\right\}
	-|\alpha_k|^2
	\sum_{j\ne k}
	\left|\bar{\bm h}_k^H\bm w_{\mathrm B,j}\right|^2
	\nonumber\\
	&-|\alpha_k|^2\left(
	\sigma_{\mathrm R}^2
	\left\| \bm h_{2,k}^{H}\bm W_{\mathrm R}\right\|^2
	+\sigma_{\mathrm U,k}^2
	\right),
	\label{eq:gk_WB}
\end{align}
which is a concave quadratic function of $\bm W_{\mathrm B}$. Accordingly, the corresponding subproblem is formulated as 
\begin{subequations}\label{prob:WB}
	\begin{align}
		\max_{\boldsymbol W_{\mathrm B},\gamma}\quad
		&\gamma
		\label{prob:WB_obj}\\
		\text{s.t.}\quad
		&\hat g_k(\bm W_{\mathrm B})
		\geq\gamma,\ \forall k\in\mathcal K,
		\label{prob:WB_sinr}\\
		&\eqref{prob:main_pb},\ 
		\eqref{prob:main_pr}.
		\label{prob:WB_power}
	\end{align}
\end{subequations}
Problem \eqref{prob:WB} is a convex quadratically constrained quadratic program (QCQP). Its globally optimal solution can be obtained using standard tools such as CVX. 

\subsection{Relay Precoding Matrix Update}
For fixed $\bm\alpha$, $\bm F$, and $\bm W_{\mathrm B}$, $\bm W_{\mathrm R}$ is embedded in several matrix products in the transformed SINR and relay power constraints, so the structure of the resulting subproblem is not immediately apparent. We therefore introduce the vectorized variable $\bm w_{\mathrm R}\triangleq\operatorname{vec}(\bm W_{\mathrm R})\in\C^{N^2}$. Applying the identity $\operatorname{vec}(\bm A\bm B\bm C)=(\bm C^T\otimes\bm A)\operatorname{vec}(\bm B)$ with $\bm A=\bm h_{2,k}^{H}$, $\bm B=\bm W_{\mathrm R}$, and $\bm C=\bm H_1\bm w_{\mathrm B,j}$ gives $\bm h_{2,k}^{H}\bm W_{\mathrm R}\bm H_1\bm w_{\mathrm B,j}=\bm v_{k,j}^{H}\bm w_{\mathrm R}$, where $\bm v_{k,j}^{H}\triangleq\left( \bm H_1\bm w_{\mathrm B,j}\right)^T\otimes\bm h_{2,k}^{H}$. Similarly, setting $\bm C=\bm I_N$ yields $\operatorname{vec}(\bm h_{2,k}^{H}\bm W_{\mathrm R}) =\bm U_k^{H}\bm w_{\mathrm R}$, where $\bm U_k^{H}\triangleq\bm I_N\otimes\bm h_{2,k}^{H}$. Taking the squared Euclidean norm on both sides gives  $\left\|\bm h_{2,k}^{H}\bm W_{\mathrm R}\right\|^2 =\left\| \bm U_k^{H}\bm w_{\mathrm R}\right\|^2 =\bm w_{\mathrm R}^{H}\bm U_k\bm U_k^{H}\bm w_{\mathrm R}$. Substituting the above relations into $\hat g_k$ gives the following equivalent form: \looseness=-1
\begin{align}
	\hat g_k(\bm w_{\mathrm R})
	={}&2\operatorname{Re}\!\left\{
	\alpha_k^*\bm v_{k,k}^{H}\bm w_{\mathrm R}
	\right\}
	-|\alpha_k|^2\sum_{j\ne k}
	\left|\bm v_{k,j}^{H}\bm w_{\mathrm R}\right|^2
	\nonumber\\
	&-|\alpha_k|^2\sigma_{\mathrm R}^2
	\bm w_{\mathrm R}^{H}\bm U_k\bm U_k^{H}\bm w_{\mathrm R}
	-|\alpha_k|^2\sigma_{\mathrm U,k}^2,
	\label{eq:gk_WR}
\end{align}
which is a concave quadratic function of $\bm w_{\mathrm R}$. Moreover, the relay transmit power can be written as $\bm w_{\mathrm R}^{H}(\bm R_{\mathrm{in}}^T\otimes\bm I_N)\bm w_{\mathrm R}$. Combining the above expressions yields the following subproblem:
\begin{subequations}\label{prob:WR}
	\begin{align}
		\max_{\boldsymbol w_{\mathrm R},\gamma}\quad
		&\gamma
		\label{prob:WR_obj}\\
		\text{s.t.}\quad
		&\hat g_k(\bm w_{\mathrm R})
		\geq\gamma,\quad \forall k\in\mathcal K,
		\label{prob:WR_sinr}\\
		&\bm w_{\mathrm R}^{H}
		(\bm R_{\mathrm{in}}^T\otimes\bm I_N)
		\bm w_{\mathrm R}
		\leq P_{\mathrm R,\max},
		\label{prob:WR_power}
	\end{align}
\end{subequations}
which is a convex QCQP and can be optimally solved via CVX. The updated $\bm W_{\mathrm R}$ is obtained by reshaping $\bm w_{\mathrm R}$. 

\subsection{Overall Algorithm}

Starting from a feasible initialization, the proposed AO-MFW algorithm alternately updates $\bm F$, $\bm W_{\mathrm B}$, and $\bm W_{\mathrm R}$ using the methods described above. Before each block update, $\bm\alpha$ is refreshed using the latest design variables. By construction, every block update preserves feasibility and does not decrease the objective value. The finite power budgets further impose an upper bound on the achievable minimum SINR. Hence, the objective sequence generated by the AO iterations is guaranteed to converge.

The computational burden of each AO iteration is dominated by the three block updates. For fixed numbers of scattering paths, the RA pointing update requires $\mathcal O(I_{\mathrm{MFW}}C_{\mathrm F})$ operations, where $C_{\mathrm F}\triangleq I_{\mathrm A}KN(M_{\mathrm B}+N+K)$. For the remaining two variable blocks, solving the QCQPs \eqref{prob:WB} and \eqref{prob:WR} using an interior-point method with solution accuracy $\epsilon$ incurs $\mathcal O(C_{\mathrm B}\ln(1/\epsilon))$ and $\mathcal O(C_{\mathrm R}\ln(1/\epsilon))$, respectively, where $C_{\mathrm B}\triangleq(KM_{\mathrm B}+1)^2(KM_{\mathrm B}+K+3)^{1.5}$ and $C_{\mathrm R}\triangleq(N^2+1)^2(N^2+K+2)^{1.5}$ \cite{2014_K.wang_complexity}. Therefore, the per-AO-iteration complexity is about $\mathcal O\!\left(I_{\mathrm{MFW}}C_{\mathrm F}+(C_{\mathrm B}+C_{\mathrm R})\ln(1/\epsilon)\right)$.

\begin{remark}
	\label{rem:direct_multi}
	\rm When the direct BS--user links are non-negligible, each user can use interference-rejection combining (IRC) to combine the signals received over the direct and relay links. The resulting joint design problem can still be reformulated via the quadratic transform and addressed within a modified AO-MFW framework. Specifically, the desired-signal and interference terms at the IRC output are used to derive revised RA gradients for the MFW update of $\bm F$ and to reformulate the subproblems for updating $\bm W_{\mathrm B}$ and $\bm W_{\mathrm R}$. 
\end{remark}

\section{Simulation Results}\label{sec:simulation}

This section evaluates the proposed algorithms through numerical simulations. The carrier frequency is set to $3.5$ GHz, corresponding to a wavelength of $\lambda=0.0857$ m. Unless otherwise specified, the RA array consists of $N_x\times N_y=4\times4$ elements with an inter-element spacing of $\lambda/2$, and the maximum
allowable zenith angle is set to $\theta_{\max}=\pi/3$. 

\subsection{Single-User Relay System}

We first consider a single-user relay system, where both the BS and the user are equipped with a single antenna. For performance comparison, the following schemes are considered: (i) \textbf{Proposed MFW:} the method proposed in Section~\ref{sec:single_user}. (ii) \textbf{Common orientation:} a constrained version of the proposed MFW scheme where all RAs share the same orientation, i.e., $\bm f_1=\cdots=\bm f_N$. The $N$ pointing vectors therefore reduce to a single optimization variable. (iii) \textbf{Half split:} a heuristic scheme where the RAs are equally divided into two groups, with one group oriented toward the BS and the other toward the user. (iv) \textbf{Random orientation:} a benchmark scheme with randomly generated RA orientations, whose performance is further averaged over 100 independent orientation realizations for each channel and geometry realization. (v) \textbf{Fixed orientation:} a benchmark scheme where all RAs are fixed to point toward the positive $z$-axis. (vi) \textbf{Isotropic antenna:} a conventional benchmark in which every antenna has unit gain in all directions. All reported results are averaged over 300 independent channel and geometry realizations.

\begin{figure}[!t]
	\vspace{-3mm}
	\centering
	\includegraphics[scale=0.6]{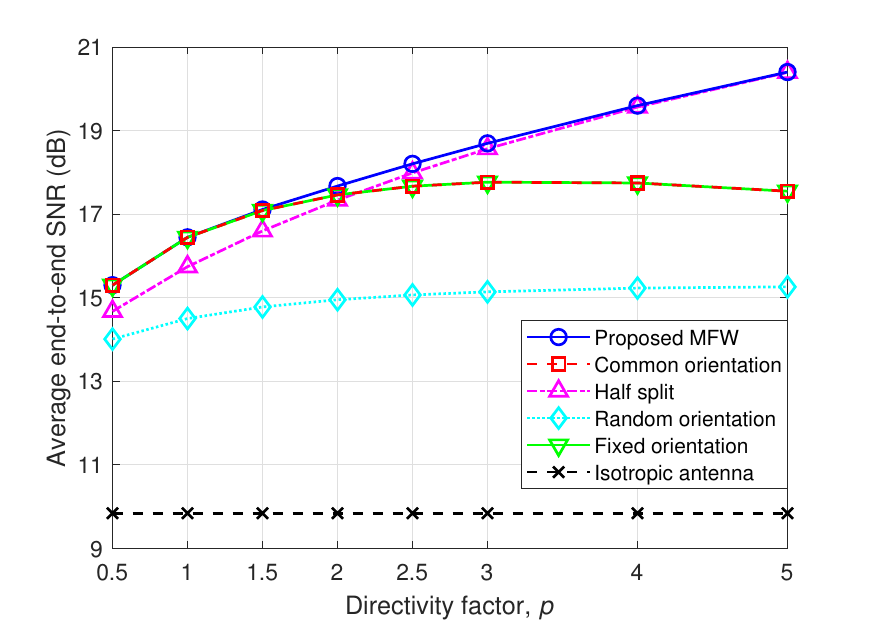}
	\caption{Average end-to-end SNR versus directivity factor under a symmetric far-field LoS geometry.} 
	\label{fig:SNR_vs_p_LoS}
\end{figure}

\subsubsection{Symmetric Far-Field LoS Geometry}
A symmetric far-field LoS geometry is considered, where the BS and user have identical distances uniformly generated from $[100,120]$ m and identical zenith angles uniformly generated from $[\pi/6,\pi/3]$, with an azimuth angle difference of $\pi$. The transmit and noise powers are set as $P_{\rm B,\max}=P_{\rm R,\max}=5$ dBm and $\sigma_{\rm R}^2=\sigma_{\rm U}^2=-80$ dBm, respectively. 

Fig.~\ref{fig:SNR_vs_p_LoS} depicts the average end-to-end SNR versus the directivity factor $p$ under this geometry. It is first observed that the proposed MFW algorithm consistently achieves the best performance, with its gain over the random orientation and isotropic antenna schemes expanding as $p$ increases. More importantly, the common orientation scheme matches the proposed MFW algorithm for small $p$, whereas the half split scheme becomes superior and approaches the proposed MFW algorithm as $p$ increases. This trend corroborates Proposition~\ref{prop:splitting}, demonstrating the shift from common pointing to antenna splitting. We also note that the common orientation and fixed orientation schemes yield identical performance, as pointing toward the positive $z$-axis balances the directional gains of the two symmetric hops.

\begin{figure}[!t]
	\centering
	\includegraphics[scale=0.6]{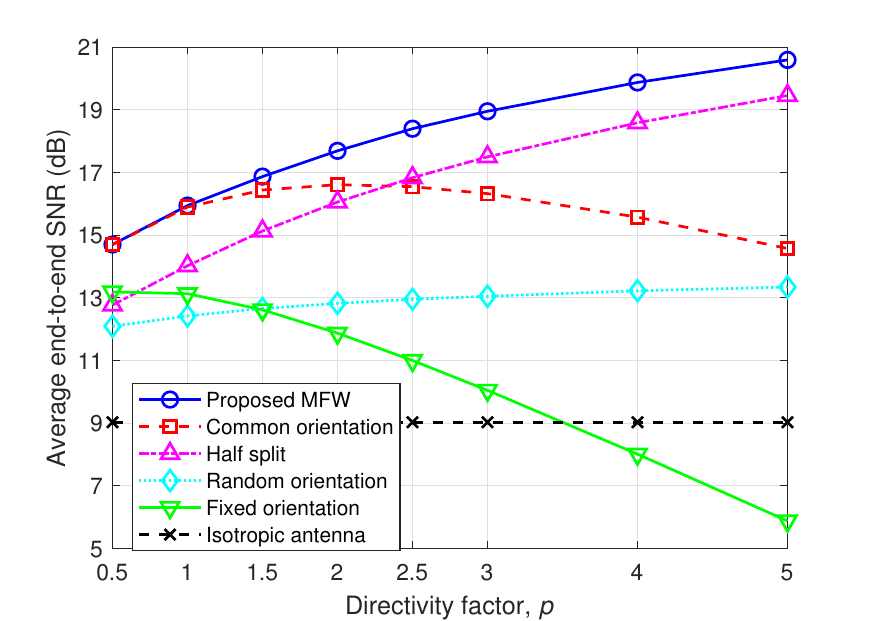}
	\caption{Average end-to-end SNR versus directivity factor under a general asymmetric channel geometry.} 
	\label{fig:SNR_vs_p_NLoS}
	\vspace{-2mm}
\end{figure}

\begin{figure}[!t]
	\centering
	\includegraphics[scale=0.6]{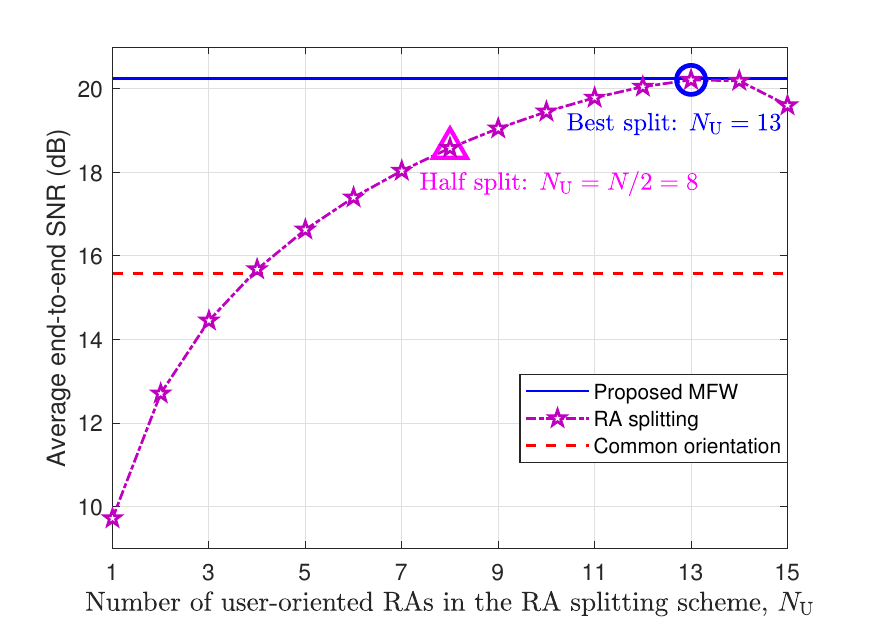}
	\caption{Average end-to-end SNR versus number of user-oriented RAs in the RA splitting scheme.} 
	\label{fig:SNR_vs_Nu}
	\vspace{-2mm}
\end{figure}

\subsubsection{Asymmetric Multipath Channel Geometry}
An asymmetric multipath channel geometry is further considered. The BS--relay and relay--user distances are uniformly generated from $[90,110]$ m and $[150,190]$ m, respectively. Their zenith angles are uniformly generated from $[5\pi/36,\pi/4]$ and $[2\pi/9,\pi/3]$, while the corresponding azimuth angles are generated from $[-5\pi/18,\pi/9]$ and $[11\pi/18,19\pi/18]$, respectively. We set $Q_1=Q_2=5$, with the scatterers randomly distributed around the midpoint of the corresponding link. The transmit powers are set as $P_{\rm B,\max}=P_{\rm R,\max}=5$ dBm, while $\sigma_{\rm R}^2=-90$ dBm and $\sigma_{\rm U}^2=-80$ dBm.

Fig.~\ref{fig:SNR_vs_p_NLoS} shows the average end-to-end SNR versus the directivity factor $p$ under the general asymmetric channel geometry. Several observations can be made. First, the proposed MFW algorithm maintains its superiority over all benchmark schemes across the entire range of $p$. Second, the fixed orientation scheme suffers a significant performance degradation as $p$ increases and even falls below the isotropic antenna scheme for $p\geq 4$. This is because the narrower beamwidth at higher directivity exacerbates the gain loss caused by directional misalignment. Third, the common orientation scheme initially improves and then degrades with $p$, since a single shared orientation becomes increasingly restrictive in accommodating the different propagation directions of the two hops as the beam narrows. Finally, the half split scheme improves monotonically with $p$ and becomes the best benchmark at large $p$, but still remains below the proposed MFW algorithm because its fixed equal partition and prescribed pointing directions cannot fully adapt to the asymmetric multipath channels.

To further examine the antenna-splitting behavior in asymmetric multipath channels, we consider a generalized \textbf{RA splitting} scheme, where $N_{\rm U}$ RAs are oriented toward the user and the remaining $N-N_{\rm U}$ RAs toward the BS. Fig.~\ref{fig:SNR_vs_Nu} shows the average end-to-end SNR versus $N_{\rm U}$ when $p = 4$. We observe that the RA splitting scheme achieves its peak SNR at $N_{\rm U}=13$, closely matching the proposed MFW algorithm, whereas the half split scheme ($N_{\rm U}=N/2=8$) suffers a noticeable performance gap. This indicates that large $p$ still favors antenna splitting under asymmetric multipath channels, while the preferred partition becomes highly unbalanced, extending the insight of Proposition~\ref{prop:splitting} beyond the symmetric LoS case.

\begin{figure}[!t]
	\centering
	\subfigure[]{\label{fig:SINR_vs_p}
		\includegraphics[scale=0.6]{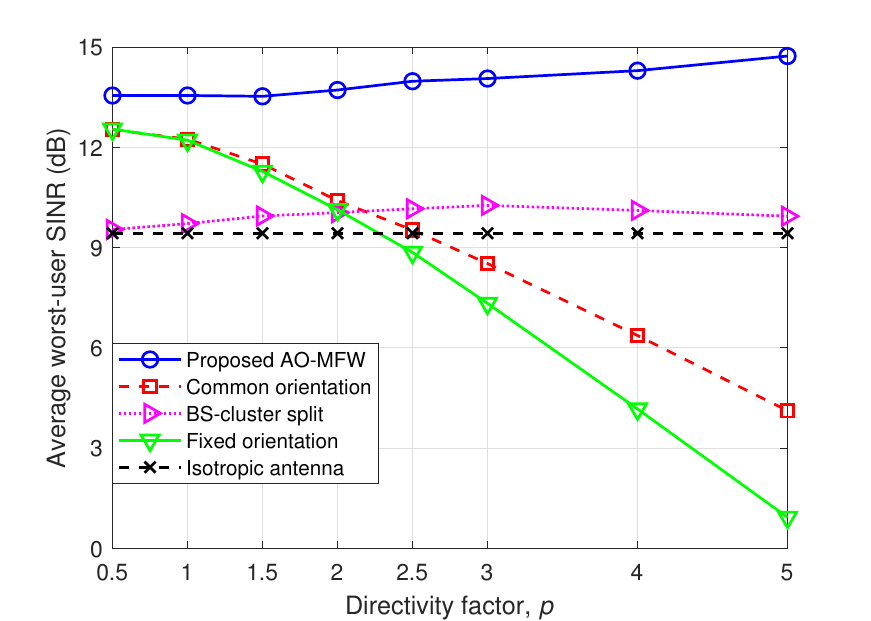}}\vspace{-2mm}
	\subfigure[]{\label{fig:orientation_stats_vs_p}
		\includegraphics[scale=0.6]{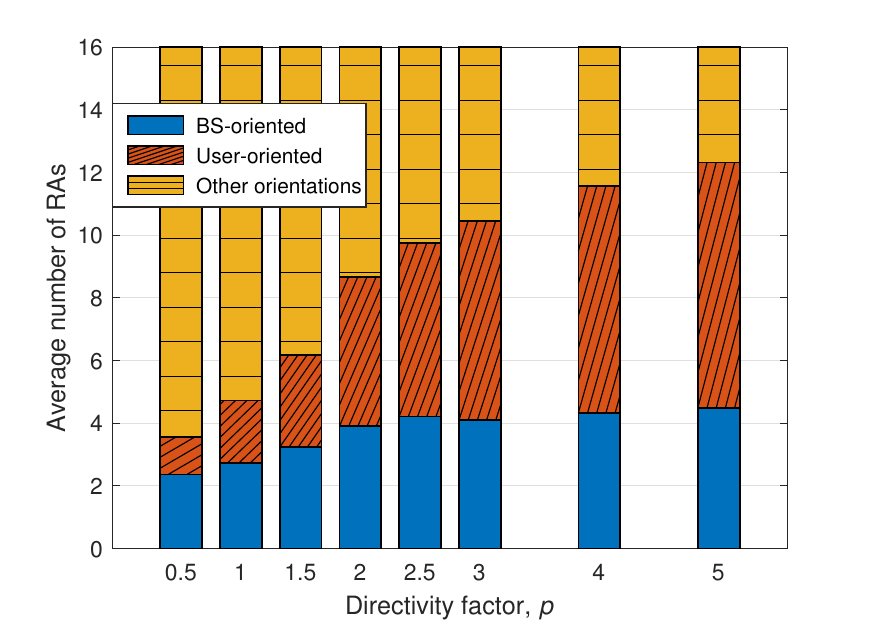}}
	\caption{Impact of the directivity factor $p$ in the multiuser scenario: (a) average worst-user SINR achieved by different schemes; (b) average numbers of BS-oriented, user-oriented, and other RAs under the proposed algorithm.}
	\label{fig:SINR_orientation_stats_vs_p}
	\vspace{-1mm}
\end{figure}

\subsection{Multiuser Relay System}

We next evaluate the proposed design for the multiuser relay system, where the BS is equipped with a $4\times4$ UPA with an inter-element spacing of $\lambda/2$. The relay and BS are located at $[0,0,0]^T$ m and $[-40,0,40]^T$ m, respectively. The horizontal coordinates of the users are randomly distributed within a circular area of radius $25$ m centered at $[50,0]^T$ m, while their $z$-coordinates are uniformly generated from $[25,40]$ m. Moreover, we set $Q_1=Q_2=5$, $P_{\rm B,\max}=P_{\rm R,\max}=15$ dBm, and $\sigma_{\rm R}^2=\sigma_{\mathrm U,k}^2=-80$ dBm, $\forall k \in\mathcal K$.

Compared with the single-user case, the proposed MFW algorithm is replaced by the \textbf{proposed AO-MFW} algorithm introduced in Section~\ref{sec:multiuser}. Moreover, to accommodate the multiple user directions, we replace the half split scheme with a \textbf{BS-cluster split} scheme, where $N_{\rm B}=\mathrm{round}(N/(K+1))$ RAs are oriented toward the BS and the remaining RAs toward the user-cluster center. Apart from these modifications, schemes with the same names as in the previous subsection retain their previous definitions.

\subsubsection{Impact of Directivity Factor} Fig.~\ref{fig:SINR_orientation_stats_vs_p} investigates the impact of the directivity factor $p$ for $K=3$. As shown in Fig.~\ref{fig:SINR_vs_p}, the proposed AO-MFW algorithm consistently achieves the highest worst-user SINR, which steadily improves as $p$ increases. In contrast, both the common orientation and fixed orientation schemes deteriorate markedly at large $p$ due to the severe directional mismatch caused by narrower beams. The BS-cluster split scheme initially benefits from higher directivity, but slightly degrades at large $p$ because pointing toward the cluster center fails to cover individual user and multipath directions as the beams narrow. Interestingly, the isotropic antenna scheme remains robust and outperforms both the common orientation and fixed orientation schemes at large $p$, owing to its direction-independent gain that avoids beam misalignment. 

To further reveal the orientation structure of the proposed algorithm, Fig.~\ref{fig:orientation_stats_vs_p} shows the average numbers of BS-oriented, user-oriented, and other RAs. Specifically, for the $n$-th RA, let $\psi_{{\rm B},n}$ and $\psi_{{\rm U},n}$ denote its angular offsets from the BS direction and the closest user direction, respectively. With $\delta=\pi/6$, it is classified as BS-oriented if $\psi_{{\rm B},n}\leq\delta$ and $\psi_{{\rm U},n}-\psi_{{\rm B},n}\geq\delta$, user-oriented analogously, and otherwise as having other orientations. As $p$ increases, the number of RAs with other orientations decreases, while more RAs align with either the BS or users, revealing a clearer antenna-splitting structure at higher directivity. The predominance of user-oriented RAs at high $p$ reflects the need to accommodate multiple distributed user directions in the second hop, whereas the first hop only involves a single BS direction.

\begin{figure}[!t]
	\centering
	\includegraphics[scale=0.6]{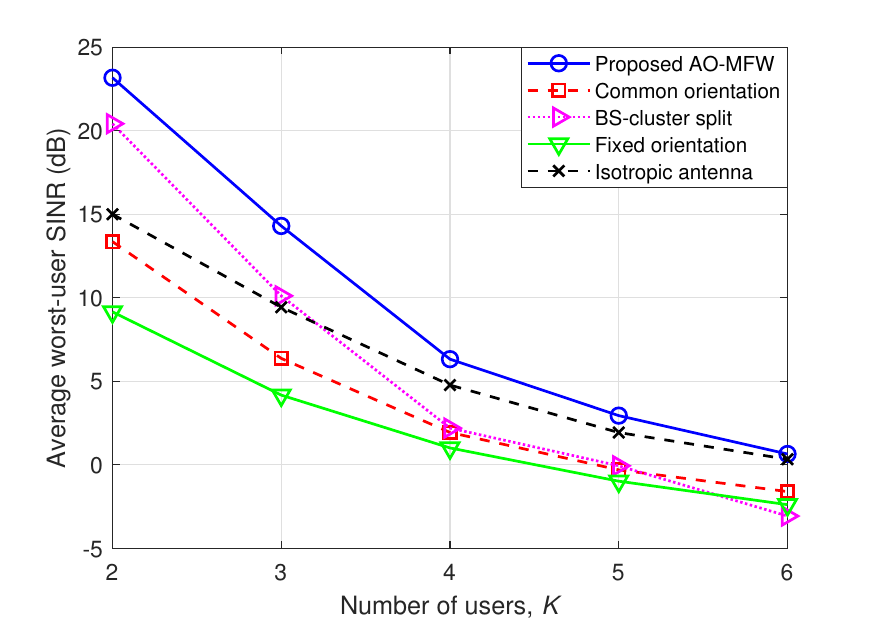}
	\caption{Average worst-user SINR versus number of users.} 
	\label{fig:SINR_vs_K}
\end{figure}

\begin{table}[t]
	\centering
	\caption{Performance comparison for different system configurations under high load ($K=6$). }
	\label{tab:high_load}
	\begin{tabular}{l c c c c c c}
		\hline
		\headcellleft{Configuration}
		& \headcell{$N$}
		& \headcell{$p$}
		& \headcell{$\theta_{\max}$}
		& \headcell{Prop.\\(dB)}
		& \headcell{Iso.\\(dB)}
		& \headcell{Prop.$-$Iso.\\(dB)}
		\\    
		\hline\noalign{\vskip 1.5pt}
		Default setting
		& 16 & 4 & $\pi/3$
		& 0.646 & 0.334 & 0.312 \\
		
		Larger $N$
		& 25 & 4 & $\pi/3$
		& 2.865 & 1.911 & 0.954 \\
		
		Larger $N$
		& 36 & 4 & $\pi/3$
		& 4.950 & 3.227 & 1.723 \\
		
		Smaller $p$
		& 16 & 3 & $\pi/3$
		& 1.120 & 0.334 & 0.786 \\
		
		Smaller $p$
		& 16 & 2 & $\pi/3$
		& 1.327 & 0.334 & 0.993 \\
		
		Larger $\theta_{\max}$
		& 16 & 4 & $5\pi/12$
		& 0.881 & 0.334 & 0.547 \\
		
		Larger $\theta_{\max}$
		& 16 & 4 & $\pi/2$
		& 0.999 & 0.334 & 0.665 \\
		\hline
	\end{tabular}
\end{table}

\subsubsection{Impact of Number of Users and Performance Under High Load}
Fig.~\ref{fig:SINR_vs_K} shows the average worst-user SINR versus the number of users $K$ with $p=4$. As expected, the SINR of all schemes decreases as $K$ increases, due to the more stringent multiuser interference management and fairness requirements. Nevertheless, the proposed AO-MFW algorithm consistently achieves the best performance, demonstrating its ability to adapt the individual RA orientations to increasingly diverse user directions. Among the benchmark schemes, the BS-cluster split scheme performs well for small $K$ but degrades rapidly as $K$ increases, since its fixed partition leaves fewer RAs oriented toward the BS while cluster-center pointing becomes less effective in accommodating the growing number of user directions. Interestingly, the isotropic antenna scheme becomes the best benchmark at relatively large $K$, as its direction-independent gain avoids the increasing orientation mismatch suffered by the constrained directional schemes. 

Motivated by the high-load case in Fig.~\ref{fig:SINR_vs_K}, Table~\ref{tab:high_load} further compares the proposed AO-MFW (Prop.) and isotropic antenna (Iso.) schemes at $K=6$ under different settings of $N$, $p$, and $\theta_{\max}$. As can be observed, increasing $N$ or $\theta_{\max}$ provides greater spatial or orientation flexibility, thereby improving the worst-user SINR and enlarging the gain over isotropic antennas to $1.723$ dB at $N=36$. More notably, the effect of $p$ differs from that observed at $K=3$: decreasing $p$ from $4$ to $2$ at $K=6$ improves the worst-user SINR of the proposed algorithm from $0.646$ dB to $1.327$ dB and increases its gain over isotropic antennas to $0.993$ dB. This is because the broader antenna pattern at smaller $p$ can better accommodate multiple user directions under high load, despite the reduced peak directional gain. Overall, Table~\ref{tab:high_load} shows that the performance advantage of the proposed design under high load can be considerably enhanced by increasing the array size, using a smaller directivity factor, or allowing a wider rotation range.

\begin{figure}[!t]
	\centering
	\includegraphics[scale=0.6]{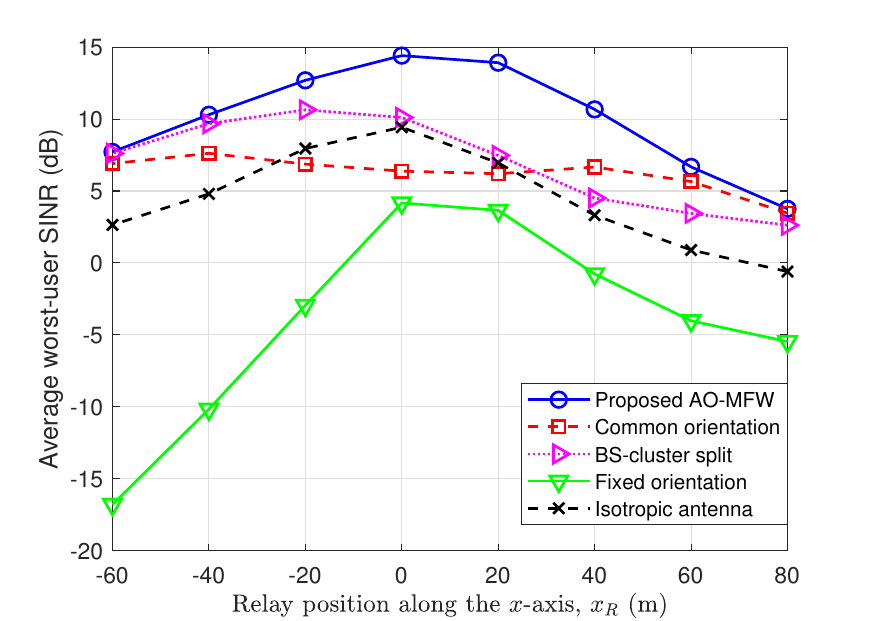}
	\caption{Average worst-user SINR versus relay position along the $x$-axis.} 
	\label{fig:SINR_vs_xR}
	\vspace{-1mm}
\end{figure}

\subsubsection{Impact of Relay Position}

Finally, we investigate the impact of relay position, which jointly affects the two-hop propagation distances and RA pointing directions. Specifically, we set $K=3$ and $p=4$, and move the relay along the $x$-axis as $[x_{\rm R},0,0]^T$~m with $x_{\rm R} \in [-60,80]$~m. As shown in Fig.~\ref{fig:SINR_vs_xR}, the proposed AO-MFW algorithm consistently achieves the highest worst-user SINR. Both the proposed algorithm and the isotropic antenna scheme perform best at intermediate relay positions, where the two-hop path losses are relatively balanced, with the proposed design maintaining a clear lead through orientation adaptation. At the two ends of the considered range, the performance gaps among the proposed, common orientation, and BS-cluster split schemes narrow because one hop becomes the dominant bottleneck and the BS and users appear in more similar directions. Specifically, as the relay approaches the BS, the relay--user hop becomes the dominant bottleneck, making the BS-cluster split scheme more favorable; conversely, as the relay approaches the users, the BS--relay hop dominates and the common orientation scheme becomes more favorable. In contrast, the fixed orientation scheme degrades sharply away from the intermediate region due to severe directional misalignment. Overall, relay placement with relatively balanced two-hop propagation conditions is generally preferable.

\section{Conclusion}\label{sec:conclusion}

In this paper, we investigated an RA-enhanced half-duplex AF relaying system and jointly optimized the BS precoding, relay precoding, and RA pointing matrices to maximize the minimum SINR among the users. To address the resulting non-convex problem, we first applied an MFW method to the single-user scenario and derived a critical directivity threshold separating common pointing from antenna splitting under a symmetric far-field LoS geometry. Building on this MFW procedure, we applied the quadratic transform to the general multiuser problem and solved the resulting auxiliary-variable formulation via AO with a safeguarded MFW update. Numerical results verified the derived threshold, showed that the proposed algorithms achieve the best performance among all considered schemes, and revealed adaptive RA pointing structures for balancing the two hops, the advantage of a moderate directivity factor under heavy user loads, and the benefit of relay placement with relatively balanced two-hop propagation conditions. \looseness=-1

\appendix[Proof of Proposition~\ref{prop:splitting}]\label{app:proof_splitting}
Without loss of generality, take the common propagation plane to be the $x$--$z$ plane. An RA pointing at angle $\theta$ has pointing vector $\bm f(\theta)=[\sin\theta,0,\cos\theta]^T$, while the first- and second-hop propagation directions are $\bm u_1=[\sin\varphi,0,\cos\varphi]^T$ and $\bm u_2=[-\sin\varphi,0,\cos\varphi]^T$, respectively. Hence, $\bm f^T(\theta)\bm u_1=\cos(\theta-\varphi)$ and $\bm f^T(\theta)\bm u_2=\cos(\theta+\varphi)$. According to \eqref{eq:gain_pattern}, the corresponding orientation-dependent power-gain factors are $a_p(\theta)\triangleq\left[ \cos(\theta-\varphi)\right] _+^{2p}$ and $b_p(\theta)\triangleq\left[ \cos(\theta+\varphi)\right]_+^{2p}$, respectively. Define their aggregate values as $A(\bm F)\triangleq\sum_{n=1}^{N}a_p(\theta_n)$ and $B(\bm F)\triangleq\sum_{n=1}^{N}b_p(\theta_n)$. The LoS channel model then gives $y_1(\bm F)=\eta A(\bm F)$ and $y_2(\bm F)=\eta B(\bm F)$, where $\eta\triangleq\beta_0G_0r_{\mathrm c}^{-2}$ and $r_{\mathrm c}$ denotes the common BS--relay and relay--user distance.

Let $\widetilde{\Gamma}\triangleq\eta\Gamma_{\mathrm B}=\eta\Gamma_{\mathrm R}$ and $S(\bm F)\triangleq A(\bm F)+B(\bm F)$. Substituting $y_1(\bm F)=\eta A(\bm F)$ and $y_2(\bm F)=\eta B(\bm F)$ into \eqref{eq:single_snr} and applying the arithmetic--geometric mean inequality yield
\begin{equation}
	\gamma_{\mathrm{su}}^\star(\bm F)
	=
	\frac{\widetilde{\Gamma}^2A(\bm F)B(\bm F)}
	{\widetilde{\Gamma}S(\bm F)+1}
	\leq
	\frac{\widetilde{\Gamma}^2S^2(\bm F)}
	{4\left(\widetilde{\Gamma}S(\bm F)+1\right)}.
	\label{eq:symmetric_snr_bound}
\end{equation}
Equality holds if and only if $A(\bm F)=B(\bm F)$. The derivative of the upper bound with respect to $S$ is $\widetilde{\Gamma}^2S(\widetilde{\Gamma}S+2)/[4(\widetilde{\Gamma}S+1)^2]>0$ for $S>0$. Consequently, any feasible pointing configuration that maximizes $S(\bm F)$ and satisfies $A(\bm F)=B(\bm F)$ is globally optimal.

To maximize $S(\bm F)$, define the total directional gain contributed by an RA over the two hops as $g_p(\theta)\triangleq a_p(\theta)+b_p(\theta)$. Accordingly, maximizing $S(\bm F)=\sum_{n=1}^{N}g_p(\theta_n)$ reduces to maximizing $g_p(\theta)$ over $[-\theta_{\max},\theta_{\max}]$. Since $g_p(-\theta)=g_p(\theta)$, only non-negative pointing angles need to be considered. For $\theta\in(\varphi,\theta_{\max}]$, we have $a_p(\theta)=\cos^{2p}(\theta-\varphi)$ and $b_p(\theta)=\left[ \cos(\theta+\varphi)\right]_+^{2p}$. The former is strictly decreasing, whereas the latter is non-increasing and vanishes when $\theta+\varphi\geq\pi/2$. Hence, $g_p(\theta)$ is strictly decreasing over $(\varphi,\theta_{\max}]$, and no maximizer can lie in this interval. It therefore suffices to maximize $g_p(\theta)$ over $[0,\varphi]$. The behavior of $g_p(\theta)$ over this interval depends on the value of $2\varphi$ relative to $\pi/2$. We consider the following two subcases.

\emph{(a) $0<\varphi<\pi/4$:} Both gain terms are active throughout $[0,\varphi]$. Differentiating $g_p(\theta)$ gives
\begin{equation}
	\begin{aligned}
		g_p'(\theta)
		={}&2p\sin(\varphi-\theta)
		\cos^{2p-1}(\varphi-\theta)\\
		&-2p\sin(\varphi+\theta)
		\cos^{2p-1}(\varphi+\theta).
	\end{aligned}
	\label{eq:gp_derivative}
\end{equation}
Setting \eqref{eq:gp_derivative} to zero and rearranging gives the stationary equation in \eqref{eq:optimal_splitting_angle}, namely, $\frac{\sin(\varphi-\theta)}{\sin(\varphi+\theta)} = \left[\frac{\cos(\varphi+\theta)}{\cos(\varphi-\theta)}\right]^{2p-1}$. Taking logarithms on both sides then shows that $g_p'(\theta)=0$ if and only if $2p=T_\varphi(\theta)$, where $T_\varphi(\theta)\triangleq
1+\frac{\ln[\sin(\varphi-\theta)/\sin(\varphi+\theta)]}
{\ln[\cos(\varphi+\theta)/\cos(\varphi-\theta)]}$.
Moreover, the sign of $g_p'(\theta)$ agrees with that of
$2p-T_\varphi(\theta)$. It therefore remains to characterize $T_\varphi(\theta)$. The function $T_\varphi(\theta)$ is continuous over $(0,\varphi)$ and, by direct differentiation, strictly increasing, with $\lim_{\theta\to0^+}T_\varphi(\theta)=1/\sin^2\varphi$ and $\lim_{\theta\to\varphi^-}T_\varphi(\theta)=\infty$. Hence, if $p\leq1/(2\sin^2\varphi)$, then $g_p'(\theta)<0$ over $(0,\varphi)$, and the maximizer is $\theta=0$. Otherwise, $g_p'(\theta)$ changes sign exactly once from positive to negative, yielding a unique maximizer $\theta^\star\in(0,\varphi)$. Since $T_\varphi(\theta^\star)=2p$ and $T_\varphi(\theta)$ continuously and strictly increases from $1/\sin^2\varphi$ to infinity as $\theta$ ranges from $0$ to $\varphi$, $\theta^\star$ continuously increases from zero toward $\varphi$ as $p$ increases from $p_{\mathrm{th}}(\varphi)$ to infinity.

\emph{(b) $\pi/4\leq\varphi<\pi/2$:} For $\pi/4<\varphi<\pi/2$, both gain terms are active over
$\theta\in(0,\pi/2-\varphi)$, and the same $T_\varphi(\theta)$ is strictly decreasing from $1/\sin^2\varphi$ to $1$. Hence, $g_p'(\theta)$ has at most one zero in this interval and, when it exists, changes sign from negative to positive. The corresponding stationary point is therefore a local minimum. Over $[\pi/2-\varphi,\varphi]$, $b_p(\theta)=0$, while $a_p(\theta)=\cos^{2p}(\varphi-\theta)$ increases monotonically with $\theta$. For $\varphi=\pi/4$, we have $T_\varphi(\theta)=2$ over $(0,\varphi)$, so $g_p(\theta)$ is decreasing, constant, or increasing according as $p<1$, $p=1$, or $p>1$. Hence, for $\pi/4\leq\varphi<\pi/2$, a global maximum of $g_p(\theta)$ is attained at either $\theta=0$ or $\theta=\varphi$. Since $g_p(0)=2\cos^{2p}\varphi$ decreases strictly with $p$, whereas $g_p(\varphi)=1$, the two endpoints yield the same value at $p=\ln 2/[-2\ln(\cos\varphi)]$. Hence, $\theta=0$ is optimal at or below this threshold, while $\theta^\star=\varphi$ above it.

We now use the above scalar characterization to establish cases~(i) and~(ii) of Proposition~\ref{prop:splitting}. (i) For $0<p\leq p_{\mathrm{th}}(\varphi)$, $\theta=0$ maximizes $g_p(\theta)$. Hence, setting $\theta_n=0$ for all $n$ maximizes $S(\bm F)$. Since $a_p(0)=b_p(0)$, this configuration also satisfies $A(\bm F)=B(\bm F)$. It therefore attains the upper bound in
\eqref{eq:symmetric_snr_bound} and is globally optimal.  (ii) For $p>p_{\mathrm{th}}(\varphi)$, both $+\theta^\star$ and $-\theta^\star$ maximize $g_p(\theta)$, where $\theta^\star\in(0,\varphi)$ for $0<\varphi<\pi/4$ and $\theta^\star=\varphi$ for $\pi/4\leq\varphi<\pi/2$. Since $a_p(-\theta)=b_p(\theta)$ and $b_p(-\theta)=a_p(\theta)$, assigning $N/2$ RAs to each angle maximizes $S(\bm F)$ and gives $A(\bm F)=B(\bm F)=Ng_p(\theta^\star)/2$. The resulting antenna-splitting configuration therefore attains the upper bound in \eqref{eq:symmetric_snr_bound} and is globally optimal. Together with the asymptotic result established in subcase~(a), the above conclusions complete the proof. 


\bibliographystyle{IEEEtran}
\bibliography{ref}

\end{document}